\documentclass[12pt]{article}
\usepackage{amsmath, amsfonts, amssymb, amsthm, graphicx, geometry, arydshln, umoline, subfig, newtxtext, newtxmath, color, comment, soul, enumerate, bm}
\usepackage{appendix}
\usepackage{natbib}
\usepackage[colorlinks=true,citecolor=blue]{hyperref}
\usepackage{booktabs}
\usepackage[table]{xcolor}
\def\BibTeX{{\rm B\kern-.05em{\sc i\kern-.025em b}\kern-.08em
    T\kern-.1667em\lower.7ex\hbox{E}\kern-.125emX}}
\usepackage{makecell}
\usepackage{bbm}

\theoremstyle{definition}
\newtheorem{theorem}{Theorem}%[section]

\newtheorem{corollary}{Corollary}

\newtheorem{definition}{Definition}
\newtheorem{example}{Example}

\newtheorem{lemma}{Lemma}

\newcommand{\black}{\color{black}}

\title{\fontsize{17pt}{1cm}\selectfont
{A Characterization of Egalitarian and Proportional Sharing Principles: \\An Efficient Extension Operator Approach}\footnote{We thank Sylvain B\'eal, Susumu Cato, Sylvain Ferri\`eres,  St\'ephane Gonzalez,  Toru Hokari,  Nobuo Koida,  David Lowing,  Juan D. Moreno-Ternero,  Shintaro Miura,  Florian Navarro,  Hendrik Rommeswinkel,  Philippe Solal, Kevin Techker, Takashi Ui, and participants in EAGT 2024, Summer workshop 2024, RISS workshop 2025 in Kansai University, Prof. Koichi Tadenuma retirement conference, Universit\'e Marie et Louis Pasteur, Hitotsubashi University, Kwansei Gakuin Univeristy, SING 2025, University of Saint-Etienne, and Networks and Games seminars at CES for helpful comments. Nakada acknowledges  the financial support from Japan Society for the Promotion of Science KAKENHI: No.19K13651, 20KK0036, and 25K16606.
Koriyama acknowledges the financial support from Investissements d'Avenir, ANR-11-IDEX-0003/Labex Ecodec/ANR-11-LABX-0047.
All remaining errors are our own.}}
\author{
Yukihiko Funaki\thanks{School of Political Science and Economics, Waseda University. E-mail: funaki@waseda.jp}
\and
Yukio Koriyama\thanks{CREST, Ecole Polytechnique, Institut Polytechnique de Paris. E-mail: yukio.koriyama@polytechnique.edu}
\and
Satoshi Nakada\thanks{School of Management, Department of Business Economics, Tokyo University of Science. E-mail: snakada@rs.tus.ac.jp}
}

\date{\today}

\begin{document}

\maketitle

\begin{abstract}
\fontsize{10pt}{10.5pt}\selectfont
{
Some well-known solutions for cooperative games with transferable utility (TU-games), such as the Banzhaf value, the Myerson value, and the Aumann--Dr\`eze value, fail to satisfy efficiency.
Despite their desirable normative properties, this inefficiency motivates the search for a systematic method to restore efficiency while preserving their underlying normative structure.
By shifting attention from individual solutions to their meta-level transformations, this paper introduces \textit{efficient extension operators} as a general framework for restoring efficiency from arbitrary underlying solutions.
We consider novel axioms for those operators and characterize the egalitarian surplus sharing method and the proportional sharing method in a unified manner.
As applications, we demonstrate the generality of our method by developing an \textit{efficient-fair extension} of solutions for TU games with communication networks, as well as a variant for TU games with coalition structures.
}
\end{abstract}

\section{Introduction}
\subsection{Motivation and overview}
Efficiency is a fundamental requirement in a wide range of resource allocation problems, yet it often conflicts with other desirable properties. 
For instance, efficiency and equity axioms such as envy-freeness are generally incompatible in exchange economies \citep{tadenuma1996};\footnote{\cite{thomson1990} shows that envy-freeness and another equity criterion called egalitarian equivalence are incompatible in general.} efficiency and stability cannot always be jointly achieved in school choice problems \citep{galeshapley1962, abdulkadirouglusonmez2003}; and efficiency and strategy-proofness are often at odds in general social choice environments \citep{gibbard1973, satterthwaite1975, mullersatterthwaite1977}.

Similar trade-offs arise in cooperative games with transferable utility (TU-games), where allocation rules assign cardinal payoffs to players based on the worth of coalitions. 
These trade-offs are particularly important because TU-games provide a foundational framework for analyzing distributive justice, welfare comparisons, and collective decision-making in multi-agent environments.
Here, efficiency means that the total payoff equals the worth of the grand coalition. 
While canonical solutions such as the Shapley value \citep{Shapley1953}, the CIS value \citep{DriessenFunaki1991}, and the egalitarian and proportional sharing values satisfy efficiency, several widely discussed rules do not. 
For example, the Banzhaf index \citep{banzhaf1964} fails to allocate the entire value of the grand coalition, complicating the interpretation of comparative statics for voting power: an increase in the index does not necessarily correspond to an increase in actual voting power.
Similarly, the Myerson value \citep{myerson1977graphs}, though normatively appealing for incorporating communication structures as an extension of the Shapley value, can also violate efficiency.

In many applications, inefficient but normatively appealing rules are used as benchmarks.
A natural approach, therefore, is to restore efficiency while preserving as many of their original desirable properties as possible. 
This idea is used to resolve the aforementioned trade-offs, for instance, equity-adjusted efficient allocations in exchange economies \citep{suzumura1983, tadenuma2002,tadenuma2005,houytadenuma2009} or efficiency-enhancing modifications of the deferred acceptance algorithm in matching markets \citep{erdilergin2008, kesten2010, douganehlars2021, shirakawa2025}.

In TU-games, two main procedures are commonly used to restore efficiency from a given allocation rule: the equal surplus sharing (ESS) and proportional sharing (PS) methods.
ESS redistributes the surplus or deficit equally among players, while PS does so proportionally to the original allocation.
The CIS value is a special case of ESS, and these procedures have been widely applied—for instance, ESS to the Myerson value \citep{vandenbrinketal2012, bealetal2015, bealetal2016, bealetal2018} and PS to the Banzhaf value \citep{vandenBrinkvanderLaan1998SCW}.

If there is complete agreement on which solution is plausible, one can apply a specific method to the underlying solution in order to restore efficiency, as is typically assumed in \citet{funakikoriyama2025}.
However, there may exist multiple solutions that are plausible candidates. In particular, the legitimacy of the individual share $f$ is not usually determined uniquely.
To address the question of how efficiency should be restored in the presence of multiple candidates, it is therefore necessary to develop new operational methods that transform a wide class of $f$ into efficient ones.

To answer the question, we study how to design efficiency-restoring operators that can be applied to any rule.
Specifically, we introduce the novel notion of an \textit{efficient extension operator}, a mapping from allocation rules to possibly another efficient allocation rule.
This formulation allows us to shift the analytical focus from specific allocation rules to the operators that transform them, enabling a general axiomatic foundation for efficiency-restoring procedures such as the ESS and the PS.

Our new approach naturally leads to a distinction between two types of axioms: \textit{universal axioms} and \textit{situational axioms}.
Universal axioms are meta-level constraints on efficiency-restoring transformations: they govern how an operator maps input allocation rules into output rules, without reference to the specifics of the underlying game. By contrast, situational axioms are model-dependent requirements imposed on the resulting allocation rules and reflect the structural features of the underlying cooperative environment.

We introduce two \textit{universal} axioms imposed directly on operators: \emph{equal treatment} and \emph{equality for equal surplus}.
Equal treatment requires that players treated identically by the input rule remain so after the extension.
Equality for equal surplus ensures that if a player receives the same amount in two solutions, both in
terms of allocation and surplus, then the player should receive the same amount after the efficient
extension. 
These universal axioms substantially restrict the admissible forms of operators.
When combined with \textit{situational} axioms, namely, conditions that depend on the underlying game and are based on the monotonicity principles developed by \citet{funakikoriyama2025}, we provide unified characterizations of both the ESS and the PS methods as efficient extension operators.
Our results show that, under mild and transparent operator-level axioms, the universe of efficiency-restoring transformations boils down to only two canonical forms.

Our approach identifies the common mathematical structure underlying the ESS and PS methods through axioms formulated at the operator level. 
More importantly, unlike \citet{funakikoriyama2025}, our analysis does not presuppose a uniquely specified benchmark solution. 
Instead, we allow for a broad class of plausible benchmark solutions and characterize a single operator acting on this entire class.
A key conceptual distinction from the standard axiomatic approach is therefore that our object of study is an operator rather than an individual allocation rule. 
Consequently, our axioms constrain not only each transformed solution individually, but also the consistency of the transformation across different benchmark solutions. 
This operator-level viewpoint provides a unified foundation for efficiency-restoring procedures that cannot be obtained from benchmark-specific characterizations alone.
For this reason, our method is \textit{not} a direct technical extension of \citet{funakikoriyama2025}, but rather a complementary approach that generalizes the scope of efficient extension by shifting the focus from solutions to operators.

Next, to demonstrate the generality of our approach, we extend our analysis beyond standard TU-games to more enriched frameworks that incorporate additional structure. 
Specifically, we consider two prominent extended frameworks: TU-games with communication structures, as introduced by \cite{myerson1977graphs}, and TU-games with coalition structures, as studied by \cite{aumanndreze1974} and \cite{owen1977}.
These models allow for more realistic representations of interaction patterns among players--either through communication networks or through fixed coalition partitions.

In both settings, the fairness criterion of \cite{myerson1977graphs}, or its variants, plays a central role. 
To accommodate this consideration, we introduce the notion of an \textit{efficient-fair extension operator}: an operator that, in addition to ensuring efficiency of the output rule and preserving it when the input is already efficient, also maintains the fairness property of the input rule. 
This concept captures the idea of modifying an allocation rule to achieve efficiency without sacrificing fairness.
We show that our axiomatic approach to efficient extensions extends naturally to these enriched settings to provide axiomatic foundations for the ESS operator in communication games and coalition structure games.

\subsection{Related literature}

The idea of efficiently extending a baseline allocation rule in TU-games has been widely studied.
A classic example is the proportional normalization of the Banzhaf value.
The study most closely related to ours is \cite{funakikoriyama2025}, who consider a \textit{fixed} solution $f$ as a benchmark and analyze its ESS and PS extensions through novel monotonicity axioms. 
In our framework, their analysis corresponds to the case where the domain of the input rule is a singleton. 
Moreover, their technique applies only when the input rule is symmetric.
Hence, our result can be viewed as a generalization of their result by treating the operator as acting on the space of a broader class of allocation rules.

In the context of communication games, \cite{vandenbrinketal2012} introduced an efficient egalitarian extension of the Myerson value and characterized it axiomatically. 
\citet{bealetal2015, bealetal2016, bealetal2018} developed the concept of efficient fair extensions, identifying unique solutions consistent with fairness and efficiency on connected networks. 
\cite{shanetal2019} propose a proportional extension of the Myerson value that, however, sacrifices fairness. 
Similarly, \cite{huetal2019} study an ESS-type efficient extension of the Aumann--Dr\`eze value.
All of these papers adopt the standard axiomatic approach, focusing on the properties of specific extended allocation rules. 
Our notion of extension is fundamentally different from the extension framework of \citet{bealetal2015, bealetal2016, bealetal2018} and the aforementioned studies.
They investigate how a fixed solution (e.g., Myerson value in a connected network) can be extended to a richer class of games while preserving its defining properties in the small domain.
Note that this is a standard axiomatic analysis of the solutions of the games with communication networks.
In contrast, our objective is to transform the solution itself to recover efficiency properties by using axioms on the mapping of solution spaces.
Our contribution lies in treating efficient extensions as operators acting on the space of allocation rules and providing an operator-level axiomatization of such transformations.

From a technical perspective, our analysis builds upon the redistribution framework of \citet{Casajus2015}, who characterizes simple lump-sum tax redistribution rules by efficiency, monotonicity, and symmetry. 
His model is mathematically equivalent to additive TU-games, where the resulting allocation coincides with the egalitarian Shapley value \citep{Joosten1996thesis,vandenBrinketal2013,CasajusHuettner2014JET}. 
Although our proof eventually reduces to a redistribution problem of this type, the two settings differ fundamentally. 
Casajus studies redistribution rules acting directly on payoff vectors, whereas we characterize operators acting on allocation rules. 
Establishing this reduction constitutes the main additional technical step of our proof and allows his methodology to be extended to substantially richer cooperative game environments.

The rest of the paper is organized as follows.
In Section \ref{sec_model}, we provide the basic model.
In Section \ref{sec_main}, we introduce our main concept and result.
We extend our result to communication games in Section \ref{sec communication game} and coalition structure games in Section \ref{sec_coalition}.
We discuss several extensions in Section \ref{sec_conclusingremarks} and conclude the paper in Section \ref{sec_conclude}.
Omitted proofs are relegated to the Appendices.

\section{Preliminaries}\label{sec_model}

Let $N=\{1,\ldots,n\}$ denote the set of players.
Given $N$, a cooperative game with transferable utility (TU-game) is a pair $(N,v)$, where $v$ is a characteristic function $v:2^N\to\mathbb{R}$ with $v(\emptyset)=0$.
A coalition is a nonempty subset $S\subseteq N$.
Players $i,j\in N$ are said to be \emph{symmetric} in $v$ if 
$v(S\cup\{i\})-v(S)=v(S\cup\{j\})-v(S)$ for every $S\subseteq N\setminus\{i,j\}$.
For any permutation $\pi\in\Pi$ over $N$, the game $\pi v$ is defined by $\pi v(\pi S)=v(S)$ for all $S\subseteq N$, where $\pi S=\{\pi(i)\mid i\in S\}$.
Let $\mathcal{V}^N$ denote the set of all TU-games with player set $N$, and define
\[
\mathcal{V}^N_+=\{v\in\mathcal{V}^N\mid \sum_{k\in N}v(\{k\})>0\},
\]
the subset of games in which the sum of individual worths is strictly positive.

A solution is a function $\varphi:\mathcal{V}^N\to\mathbb{R}^n$ that assigns a payoff vector to each game.
Let $\mathcal{F}$ denote the set of all solutions, and define
\[
\mathcal{F}_+=\{f\in\mathcal{F}\mid f:\mathcal{V}^N_+\to\mathbb{R}^n,\ \sum_{k\in N}f_k(v)>0\text{ for all }v\in\mathcal{V}^N_+\}.
\]
For any $\varphi,\varphi'\in\mathcal{F}$ and $i\in N$, we write $\varphi_i=\varphi'_i$ if $\varphi_i(v)=\varphi'_i(v)$ for all $v\in\mathcal{V}^N$.

We now recall several standard axioms for solutions.\footnote{We use $\varphi$ as a generic element of $\mathcal{F}$, but later use $f$ to denote another solution serving as a benchmark.}

\begin{itemize}
\item[] \textbf{Efficiency (E):}
For all $v\in\mathcal{V}^N$, $\sum_{i\in N}\varphi_i(v)=v(N)$.

\item[] \textbf{Anonymity (A):}
For all $v\in\mathcal{V}^N$ and $\pi\in\Pi$, $\varphi_i(v)=\varphi_{\pi(i)}(\pi v)$.\footnote{This axiom is also referred to as \textit{symmetry} in the literature.}

\item[] \textbf{Equal Treatment (ET):}
For all $v\in\mathcal{V}^N$ and $i,j\in N$, if $i$ and $j$ are symmetric in $v$, then $\varphi_i(v)=\varphi_j(v)$.
\end{itemize}

Typical efficient solutions include the Shapley value \citep{Shapley1953}, the equal surplus sharing (ESS) value \citep{DriessenFunaki1991}, and the proportional sharing (PS) value, defined respectively by, for any $v\in\mathcal{V}^N$ (resp.\ $v\in\mathcal{V}^N_+$) and $i\in N$,
\[
Sh_i(v)=\sum_{S\subseteq N: i\in S} \frac{(|S|-1)!(n-|S|)!}{n!}\big(v(S)-v(S\setminus\{i\})\big),
\]
\[
ESS_i(v)=v(\{i\})+\frac{1}{n}\big(v(N)-\sum_{k\in N}v(\{k\})\big),
\]
\[
PS_i(v)=\frac{v(\{i\})}{\sum_{k\in N}v(\{k\})}v(N).
\]

\cite{funakikoriyama2025} generalize both the ESS and PS values by allowing a general benchmark solution $f\in\mathcal{F}$.\footnote{The ESS value is also known as the center of imputation set (CIS) value. \cite{van2009axiomatizations} introduce another class of generalizations of the ESS value.}
Given such a benchmark $f$, the \emph{egalitarian surplus sharing value with respect to $f$} (the $f$-ESS value) is defined, for all $v\in\mathcal{V}^N$ and $i\in N$, by
\[
\varphi_i(v)=ESS_i(f)(v)=f_i(v)+\frac{1}{n}\big(v(N)-\sum_{k\in N}f_k(v)\big).
\]
The stand-alone value $f_i(v)=v(\{i\})$ corresponds to the standard ESS value.\footnote{Two-step allocation rules of this kind have also been studied in general rationing problems by \cite{hougaardetal2012,hougaardetal2013SCW,hougaardetal2013AOR}; see also \cite{timonerizquierdo2016} and \cite{harless2017}.}

Analogously, \cite{funakikoriyama2025} define the \emph{proportional sharing value with respect to $f$} (the $f$-PS value) for $f\in\mathcal{F}_+$, $v\in\mathcal{V}^N_+$, and $i\in N$:
\[
\varphi_i(v)=PS_i(f)(v)=f_i(v)+\frac{f_i(v)}{\sum_{k\in N}f_k(v)}\big(v(N)-\sum_{k\in N}f_k(v)\big)
=\frac{f_i(v)}{\sum_{k\in N}f_k(v)}v(N).
\]

They introduce the following axioms, which, together with (E) and (ET), characterize the $f$-ESS and $f$-PS values for a given $f$ satisfying additional assumptions.

\begin{itemize}
\item[] \textbf{$f$-Individualistic Property for Equal Surplus ($f$-IES):}  
For all $v,w\in\mathcal{V}^N$ and $i\in N$, if
\[
v(N)-\sum_{k\in N}f_k(v)=w(N)-\sum_{k\in N}f_k(w)
\quad\text{and}\quad f_i(v)=f_i(w),
\]
then $\varphi_i(v)=\varphi_i(w)$.

\item[] \textbf{$f$-Individualistic Property for Equal Ratio ($f$-IER):}  
For all $v,w\in\mathcal{V}^N_+$ and $i\in N$, if
\[
\frac{v(N)}{\sum_{k\in N}f_k(v)}=\frac{w(N)}{\sum_{k\in N}f_k(w)}
\quad\text{and}\quad f_i(v)=f_i(w),
\]
then $\varphi_i(v)=\varphi_i(w)$.
\end{itemize}

\begin{theorem}[\citealt{funakikoriyama2025}]\label{f-ess/pd exogenous}
Suppose $f\in\mathcal{F}$ (resp.\ $f\in\mathcal{F}_+$) satisfies the following assumptions:

\begin{itemize}
\item[(1)] $f$ satisfies (A).

\item[(2)] For any $v\in\mathcal{V}^N$ (resp.\ $v\in\mathcal{V}^N_+$), there exists $c\in\mathbb{R}$ such that for any $x\in\mathbb{R}^n$ whose $i$th component $x_i$ equals either $c$ or $f_i(v)$ for all $i\in N$, there exists $w\in\mathcal{V}^N$ with $f(w)=x$.  
Moreover, if $x_i=c$ for all $i\in N$, then $w$ is symmetric.

\item[(3)] For any $v,w\in\mathcal{V}^N$ (resp.\ $v,w\in\mathcal{V}^N_+$), if $v(S)=w(S)$ for all $S\subsetneq N$, then $f(v)=f(w)$.
\end{itemize}

Then a solution $\varphi$ satisfies (E), (ET), and ($f$-IES) (resp.\ $f$-IER) if and only if it is the $f$-ESS value (resp.\ $f$-PS value).
\end{theorem}

\section{An efficient extension operator}\label{sec_main}

Suppose that the players agree upon using $f$ as a solution.
However, it does not necessarily satisfy efficiency.
Hence, there may remain a surplus or deficit after allocation according to $f$, and we must determine how to redistribute the residual $v(N)-\sum_{k \in N}f_k(v)$ among the players.
The $f$-ESS and $f$-PS values are natural candidates: both always satisfy (E), and if $f$ itself satisfies (E), then they coincide with $f$.
Therefore, we can safely use them as efficient modifications of $f$.

In this sense, the $f$-ESS and $f$-PS values can be regarded as \textit{efficient extensions} of $f$.
In practice, however, there may be no complete agreement on which solution should be regarded as plausible.
In particular, the legitimacy of the individual share $f$ is often not uniquely determined.
In such cases, it is natural to consider a broader class of plausible solutions as possible status quo allocations, denoted by $\mathcal{D}$.
This raises the need for a systematic way to justify how each solution $f \in \mathcal{D}$ should be extended to an efficient one.
To address this issue, we introduce a general framework that maps solutions into other solutions, and define an \textit{efficient extension operator} as follows.

\begin{definition}\label{def eeop}
Let $\mathcal{D} \subseteq \mathcal{F}$ be a set of solutions.
An operator $\Phi: \mathcal{D} \rightarrow \mathcal{F}$ is an efficient extension operator in $\mathcal{D}$ if, for any $f \in \mathcal{D}$, $\Phi(f)$ satisfies (E).
\end{definition}

An efficient extension operator can be interpreted as a reduced-form bargaining mechanism in the following sense.
Suppose that the players initially agree to adopt a solution concept $f$, which may fail to be efficient.
Recognizing this inefficiency, the players engage in a bargaining process aimed at modifying $f$ into an efficient solution.
During this process, they propose alternative solutions in a structured manner, possibly reflecting strategic considerations or normative principles.
This bargaining process eventually leads to a consensus on a new solution $f'$, which is efficient.
By abstracting away the details of the bargaining process—such as the protocol, strategic environment, and equilibrium concept—we model the outcome as the application of an operator $\Phi$, which maps the original solution $f$ to an efficient one $\Phi(f)=f'$.
%Thus, $\Phi$ represents a reduced form of the underlying bargaining process.

If $\mathcal{D}=\mathcal{F}$, we simply call an efficient extension operator in $\mathcal{F}$ an efficient extension operator.
For $\mathcal{D} \subseteq \mathcal{F}$, we say that an efficient extension operator in $\mathcal{D}$ is the egalitarian (resp. proportional) surplus sharing operator if $\Phi(f)$ coincides with the $f$-ESS (resp. $f$-PS) value for any $f \in \mathcal{D}$.
By definition, the $f$-ESS/PS value is the ESS/PS operator in $\mathcal{D}=\{f\}$ for some $f$.
Moreover, Theorem \ref{f-ess/pd exogenous} shows that if a solution $\varphi$ satisfies (E), (ET), and the corresponding axiom for some $f$ satisfying certain assumptions, then it is the ESS/PS operator in $\mathcal{D}=\{f\}$.

We introduced the concept of efficient extension operators because it provides a fundamentally different perspective on the foundations of allocation rules.
In particular, although one might consider that our approach is the same as repeatedly applying Theorem \ref{f-ess/pd exogenous} by \cite{funakikoriyama2025} to each solution $f \in \mathcal{D}$ one-by-one, this is not the case.
Technically speaking, Theorem \ref{f-ess/pd exogenous} characterizes an efficient solution separately for the fixed benchmark satisfying several technical assumptions.
If $f \in \mathcal{D}$ does not satisfy such technical assumptions, such as symmetry and richness conditions stated in (2) and (3) in Theorem \ref{f-ess/pd exogenous}, we cannot apply the result. 
In contrast, our object of study is a single operator acting on an entire class of benchmark solutions, including solutions that do not satisfy the assumptions required in Theorem \ref{f-ess/pd exogenous}.
To this end, we formulate operator-level axioms that relate different benchmark solutions through a common transformation rule.

Similar to the axioms for solutions, we can, with a slight abuse of terminology, define axioms for efficient extension operators as follows.

\begin{itemize}
\item[] \textbf{Equal treatment (ET)}: 
For any $f \in \mathcal{D}$ and $i, j \in N$, if $f_i=f_j$, then $\Phi_i(f)=\Phi_j(f)$.

\item[] \textbf{Equality for equal surplus (EES)}: 
For any $f, f' \in \mathcal{D}$, $v \in \mathcal{V}^N$, and $i \in N$, if $\sum_{k \in N}f_k(v)=\sum_{k \in N}f'_k(v)$ and $f_i(v)=f'_i(v)$, then $\Phi_i(f)(v)=\Phi_i(f')(v)$.
\end{itemize}

Equal treatment (ET) requires that players $i$ and $j$ should be treated equally by the modified solution $\Phi(f)$ if they are treated equally by $f$, in the sense that $f_i=f_j$.
Note that this requirement is weaker than asking $f$ itself to satisfy symmetry or equal treatment.
Thus, (ET) is not a strong restriction on an operator.
Equality for equal surplus (EES) requires that, for any player $i$, if the surplus with respect to $f$ and the individual allocation are equal to those of $f'$ at $v$, then player $i$’s allocation under $\Phi(f)$ and $\Phi(f')$ must also coincide at $v$.
Note that $v(N)-\sum_{k \in N}f_k(v)=v(N)-\sum_{k \in N}f'_k(v)$ is equivalent to the first condition of the axiom.

The above two axioms for an operator $\Phi$ can be regarded as \textit{universal axioms}, since they restrict the functional form of $\Phi$ depending on $f$, but not on $v$.
That is, they do not constrain how $\Phi(f)$ behaves across different games.
Hence, we additionally introduce \textit{situational axioms} that specify how the modified solution $\Phi(f)$ should behave with respect to a given $v$.
The following theorem characterizes the ESS and PS operators in a unified manner, differing only in the choice of the situational axiom.

\begin{theorem}\label{operator main result}
\begin{itemize}
\item[(1)] An efficient extension operator $\Phi: \mathcal{F} \rightarrow \mathcal{F}$ satisfies (ET) and (EES), and $\Phi(f)$ satisfies ($f$-IES) for any $f \in \mathcal{F}$ if and only if $\Phi(f)=ESS(f)$.

\item[(2)] An efficient extension operator $\Phi: \mathcal{F}_+ \rightarrow \mathcal{F}_+$ satisfies (ET) and (EES), and $\Phi(f)$ satisfies ($f$-IER) for any $f \in \mathcal{F}_+$ if and only if $\Phi(f)=PS(f)$.
\end{itemize}
\end{theorem}

Independence of the axioms of Theorem \ref{operator main result} can be verified as follows.
Examples \ref{counterexample_ESS} and \ref{counterexample_PD} imply that (EES) is indispensable.
An example described in Subsection \ref{sec_otherop} implies that ($f$-IES) and ($f$-IER) are indispensable and 
the same example with player-specific weights implies that (ET) is indispensable.

We briefly explain the main idea of the proofs, with formal details deferred to Appendix \ref{appendix_proof_main}.
First, by (ET) and (EES), we can show that $\Phi_i(f)$ depends only on $f_i$ and $\sum_{k \in N}f_k$, which we denote by $\hat \Phi_i(f_i, \sum_{k \in N}f_k)$.
Next, as a key step, the situational axiom implies that $\hat \Phi_i(f_i, \sum_{k \in N}f_k): \mathcal{V}^N\rightarrow \mathbb{R}^n$ depends only on $f_i(v)$ and either $v(N)-\sum_{k \in N}f_k(v)$ or $v(N)/\sum_{k \in N}f_k(v)$.
Thus, the problem reduces to identifying a function $\gamma_i: \mathbb{R}^2 \rightarrow \mathbb{R}$ such that $\gamma_i(f_i(v), v(N)-\sum_{k \in N}f_k(v))$ (resp. $\gamma_i(f_i(v), v(N)/\sum_{k \in N}f_k(v))$) defines the operator.
This reduced problem is isomorphic to the characterization of monotonic income redistribution functions in $\mathbb{R}^n$ by \cite{Casajus2015}.
Adopting his technique, we can show that $\gamma_i$ yields the ESS/PS operator, respectively.
We also note that, although \cite{Casajus2015}’s original proof does not hold for $n=2$, our argument remains valid even in that case.

To highlight that our proof cannot be obtained by directly applying the argument of \cite{Casajus2015}, let us consider the following weaker version of (EES).

\begin{itemize}
\item[] \textbf{Weak equality for equal surplus (WEES)}: 
For any $f, f' \in \mathcal{D}$ and $i \in N$, if 
\[
\sum_{k \in N}f_k=\sum_{k \in N}f'_k~\text{and}~f_i=f'_i,
\]
then $\Phi_i(f)=\Phi_i(f')$.
\end{itemize}
We show that Theorem \ref{operator main result} itself does not hold under (WEES), although an argument similar to that used by\citet{casajus2015monotonic} partially determines the functional form of $\Phi$ as follows.

Suppose that $n \ge 3$.
Let $\Phi: \mathcal{F}\rightarrow \mathcal{F}$ be an efficient extension operator that satisfies (ET) and (WEES) as follows.
By (WEES), for any $i \in N$, $\Phi_i$ only depends on $f_i$ and $\sum_{k \in N}f_k$, so that $\Phi_i$ can be written as $\Phi_i(f)\equiv \hat \Phi_i(f_i, \sum_{k \in N}f_k)$ for any $f \in \mathcal{F}$.
Then, we show that $\hat \Phi_i=\hat\Phi_j$.
To see this, let us consider
$f = (h, \dots, h, \underset{i\text{-th}}{a}, h, \dots, h)$ and $f' = (h, \dots, h, \underset{j\text{-th}}{a}, h, \dots, h)
 \in \mathcal{F}$ for some $a, h \in \{\lambda: \mathcal{V}^N\rightarrow \mathbb{R}\}$.

Since $f_l=f'_l$ for any $l \neq i,j$ and $\sum_{k \in N}f_k=\sum_{k \in N}f'_k$, by (WEES), we have $\Phi_l(f)=\Phi_l(f')$ for any $l \neq i,j$.
Similarly, by (ET), we have $\Phi_l(f)=\Phi_{l'}(f)$ for any $l,l' \neq i$, and $\Phi_l(f')=\Phi_{l'}(f')$ for any $l,l' \neq j$.
Moreover, since $\Phi(\cdot)$ satisfy (E), we have
\begin{align*}
\Phi_i(f)+(n-1)\Phi_l(f)&=\sum_{k \in N}\Phi_k(f)\\
&=\sum_{k \in N}\Phi_k(f')\\
&=\Phi_j(f')+(n-1)\Phi_l(f'),
\end{align*}
so that $\hat \Phi_i(a, \sum_{k \in N}f_k)=
\Phi_i(f)=\Phi_j(f')=\hat \Phi_j(a, \sum_{k \in N}f_k)$.

Let a functional $\Psi^{\sum_{k \in N}f_k}: \{\lambda: \mathcal{V}^N\rightarrow \mathbb{R}\} \rightarrow \{\lambda: \mathcal{V}^N\rightarrow \mathbb{R}\}$ be such that
\[
\Psi^{\sum_{k \in N}f_k}(a)\equiv \hat \Phi(a, \sum_{k \in N}f_k)-\hat \Phi(\textbf{0}, \sum_{k \in N}f_k).
\]
Next, by (WEES), we show that $\Psi^{\sum_{k \in N}f_k}$ is an additive functional.
Let us consider $f=(a, b, (h)_{-i,j})$, $f'=(a+b, \textbf{0}, (h)_{-i,j}) \in \mathcal{F}$ for some $a,b, h \in \{\lambda: \mathcal{V}^N\rightarrow \mathbb{R}\}$, where $\textbf{0}(v)=0$ for any $v \in \mathcal{V}^N$.
By the same argument as above,  we have the following functional equation
\[
\hat \Phi(a+b, \sum_{k \in N}f_k)-\hat \Phi(\textbf{0}, \sum_{k \in N}f_k)=\Bigl(\hat\Phi(a, \sum_{k \in N}f_k)-\hat\Phi(\textbf{0}, \sum_{k \in N}f_k) \Bigr)+\Bigl(\hat\Phi(b, \sum_{k \in N}f_k)-\hat\Phi(\textbf{0}, \sum_{k \in N}f_k) \Bigr),
\]
which shows that $\Psi^{\sum_{k \in N}f_k}$ is additive.
By (E) of $\Phi(\cdot)$, for any $f \in \mathcal{F}$ and $v$, we have
\begin{align*}
v(N)&=\sum_{i \in N}\Phi_i(f)(v)\\&=\sum_{i \in N}\hat \Phi_i(f_i, \sum_{k \in N}f_k)(v)\\
&=\sum_{i \in N}\Psi^{\sum_{k \in N}f_k}(f_i)(v)+n\hat\Phi(\textbf{0}, \sum_{k \in N}f_k)(v)\\
&=\Psi^{\sum_{k \in N}f_k}\Bigl( \sum_{k \in N}f_k \Bigr)(v)+n\hat\Phi(\textbf{0}, \sum_{k \in N}f_k)(v),
\end{align*}
which implies that
\begin{equation*}\label{ESS_formula}
\Phi_i(f)(v)=\Psi^{\sum_{k \in N}f_k}(f_i)(v)+\frac{1}{n}\Bigl(v(N)- \Psi^{\sum_{k \in N}f_k}\Bigl( \sum_{k \in N}f_k \Bigr)(v)\Bigr), \forall i \in N.
\end{equation*}

Note that the above arguments also hold by replacing $\mathcal{F}$ with $\mathcal{F}_+$ and $\mathcal{V}^N$ with $\mathcal{V}^N_+$, respectively.
The above discussion shows that $\Phi$ is an additive operator.

Under some conditions, we can ensure that $\Psi^{\sum_{k \in N}f_k}$ is linear. 
However, as the following examples \ref{counterexample_ESS} and \ref{counterexample_PD} suggest, linearity is not sufficient to conclude that $\Phi$ is the ESS/PS operator with ($f$-IES)/($f$-IER), respectively, which is contrast to \cite{Casajus2015}.
Therefore, Theorem \ref{operator main result} does not hold under (WEES).

\begin{example}\label{counterexample_ESS}
Fix $w \in \mathcal{V}^N$.
For any $f \in \mathcal{F}$, $i \in N$, and $v \in \mathcal{V}^N$, define
\[
\Phi_i(f)(v)=ESS_i(f)(v)+f_i(w)-\frac{\sum_{k \in N}f_k(w)}{n}.
\]

This operator satisfies (WEES) and (ET).
Moreover, since $w$ is fixed, the term $f_i(w)-\sum_{k \in N}f_k(w)/n$ depends only on $f_i$ and $\sum_{k \in N}f_k$.
Hence, $\Phi(f)$ satisfies ($f$-IES) for any $f \in \mathcal{F}$.
However, this operator does not satisfy (EES).

To see this, consider two solutions $f,f' \in \mathcal{F}$ defined by
$f(v)=(v(\{i\}))_{i \in N}$ and $f'(v)=(v(\{1\}), 0, \ldots, 0)$.
Let $\tilde v \in \mathcal{V}^N$ satisfy $\sum_{k \neq 1}\tilde v(\{k\})=0$.
Then $f_1(\tilde v)=f'_1(\tilde v)$ and $\sum_{k \in N}f_k(\tilde v)=\sum_{k \in N}f'_k(\tilde v)$.
If we choose $w$ such that $\sum_{k \neq 1}w(\{k\}) \neq 0$, then
\[
\Phi_1(f)(\tilde v)-\Phi_1(f')(\tilde v)
= -\frac{\sum_{k \neq 1}w(\{k\})}{n} \neq 0,
\]
which violates (EES).
\end{example}

\begin{example}\label{counterexample_PD}
Fix $w \in \mathcal{V}_+^N$.
For any $f \in \mathcal{F}_+$, $i \in N$, and $v \in \mathcal{V}_+^N$, define
\[
\Phi_i(f)(v)=PS_i(f)(v)+f_i(w)-\frac{\sum_{k \in N}f_k(w)}{n}.
\]

This operator satisfies (WEES) and (ET).
Since $w$ is fixed, the term $f_i(w)-\sum_{k \in N}f_k(w)/n$ again depends only on $f_i$ and $\sum_{k \in N}f_k$.
Hence, $\Phi(f)$ satisfies ($f$-IER) for any $f \in \mathcal{F}_+$.
By applying the same logic as in Example \ref{counterexample_ESS} to the modified domain, we can see that this operator also violates (EES).
\end{example}

%The above discussion suggests that similar characterizations may be obtained by considering weaker variants of (EES), or by formulating monotonicity-like axioms in which certain equalities are replaced with inequalities (see, e.g., \citealt{Young1985IJGT}). 
%Establishing formal foundations based on such variants of the axiom is left for future research.

Our primary objective is to investigate how efficiency can be systematically achieved starting from a solution that is not necessarily efficient.
However, efficiency may not always be the most appropriate normative benchmark in TU-games.
For instance, when a game is not super-additive, the worth of the grand coalition $v(N)$ does not necessarily represent the maximum achievable value—smaller coalitions may collectively generate a higher total worth.
To address this issue, \citet{beal2021cohesive} introduce the concept of \textit{cohesive efficiency}, which requires that a solution achieves the maximum total worth attainable when coalitions act independently.
Our framework of efficient extension operators can be naturally adapted to incorporate cohesive efficiency as the target criterion.
A detailed discussion is provided in Appendix \ref{Appendix_cohesiveefficiency}.\footnote{We thank Philippe Solal for raising this issue.}

\section{Efficient-fair extension for solutions of communication games}\label{sec communication game}

In this section, we extend the framework developed in the previous section to the setting of communication games introduced by \citet{myerson1977graphs}. 
Our goal is to demonstrate how the basic arguments for efficient extension operators can be applied to this richer model.
Subsection \ref{subsec f-ESS communication} introduces an $f$-ESS value for communication games, in the spirit of \citet{funakikoriyama2025} and Theorem~\ref{f-ess/pd exogenous}.
Building on this foundation, we then define \textit{efficient-fair extension operators}, which combine efficiency with Myerson’s central fairness principle in this context.

\subsection{$f$-ESS value for communication games}\label{subsec f-ESS communication}

A communication network is an undirected graph on $N$ specified by a set of pairs $g \subseteq g^N = \{ \{i,j\} \mid i,j \in N, i \neq j \}$, where each $\{i,j\}$ represents a communication link between players $i$ and $j$. 
The complete network on $N$ is denoted by $g^N$, and the set of all possible networks is $\mathscr{G}^N = \{ g \mid g \subseteq g^N \}$.
For notational simplicity, we write $ij$ instead of $\{i,j\}$.

For any $g \in \mathscr{G}^N$, let $N_i(g) = \{ j \in N \mid i \neq j \text{ and } ij \in g \}$ denote the neighborhood of player $i$ in $g$.
For $ij \in g$, let $g-ij = g \setminus \{ij\}$ denote the network obtained by removing the link $ij$. 
For $S \subseteq N$, the restriction of $g$ to $S$ is $g|_S = \{ ij \in g \mid i,j \in S \}$.
A sequence $i_1, \ldots, i_k$ ($k \ge 2$) is a path from $i_1$ to $i_k$ in $g$ if $i_{h+1} \in N_{i_h}(g)$ for all $h = 1, \ldots, k-1$. 
A network $g$ is said to be connected in $S$ if there exists a path between any pair $i,j \in S$. 
If $g$ is connected in $N$, we simply say that $g$ is connected.
A coalition $S$ is a component of $g$ if $g$ is connected in $S$ and not connected in $S \cup \{i\}$ for any $i \notin S$.
For each $S \subseteq N$, let $S/g$ denote the partition of $S$ into the components of $g|_S$.

A communication game is a pair $(v,g) \in \mathcal{V}^N \times \mathscr{G}^N$. 
A solution in this framework is a mapping $\varphi : \mathcal{V}^N \times \mathscr{G}^N \to \mathbb{R}^n$, and we denote the set of all such mappings by $\mathcal{F}_{\mathscr{G}}$.
The most prominent solution in this class is the \textit{Myerson value}, defined by
\[
My_i(v,g) = Sh_i(v^g),
\]
where $v^g(S) = \sum_{T \in S/g} v(T)$ for all $S \subseteq N$.
\citet{myerson1977graphs} proved that this is the unique solution satisfying component efficiency and fairness.

\begin{itemize}
\item[] \textbf{Component efficiency (CE):} 
For any $(v,g)$ and any component $S \in N/g$, $\sum_{i \in S}\varphi_i(v,g) = v(S)$.

\item[] \textbf{Fairness at $v$ (FA-$v$):}
Fix $v \in \mathcal{V}^N$. 
For any $g \in \mathscr{G}^N$ and $ij \in g$, 
\[
\varphi_i(v,g) - \varphi_i(v,g-ij) = \varphi_j(v,g) - \varphi_j(v,g-ij).
\]

\item[] \textbf{Fairness (FA):}
$\varphi$ satisfies (FA-$v$) for every $v \in \mathcal{V}^N$.
\end{itemize}

(CE) requires that the total payoff within each connected component equals the worth created by that component—players in disconnected components do not cooperate.
(FA) requires that, when a link $ij$ is severed, the change in the payoffs of players $i$ and $j$ is identical.  
This captures Myerson’s idea that each player has equal bargaining power within their local network.\footnote{\cite{myerson1977graphs} originally only considers (FA). We introduce (FA-$v$) for our purpose to consider efficient-fair extension operators in Subsection \ref{subsec EFOp}.}

The Myerson value is not efficient unless $g$ is connected.
To restore efficiency, \citet{vandenbrinketal2012} proposed the \textit{efficient egalitarian Myerson value}, obtained by applying the ESS method to the Myerson value:
\[
EEMy_i(v,g)
= My_i(v,g) + \frac{1}{n}\left( v(N) - \sum_{k \in N} My_k(v,g) \right).
\]
They showed that this solution uniquely satisfies (E), (FA), and (FDS).

\begin{itemize}
\item[] \textbf{Efficiency (E):} 
For all $(v,g)$, $\sum_{i \in N}\varphi_i(v,g) = v(N)$.

\item[] \textbf{Fair distribution of the surplus (FDS):}\footnote{The original definition in \citet{vandenbrinketal2012} is expressed using subgames. Under (CE), the current formulation is equivalent.}
For all $(v,g)$ and all $C,C' \in N/g$,
\[
\frac{1}{|C|}\!\left( \sum_{i \in C}\varphi_i(v,g) - v(C) \right)
= \frac{1}{|C'|}\!\left( \sum_{i \in C'}\varphi_i(v,g) - v(C') \right).
\]
\end{itemize}

Following \citet{funakikoriyama2025}, we generalize the efficient egalitarian Myerson value by defining the equal-surplus sharing rule for any given solution $f \in \mathcal{F}_{\mathscr{G}}$:
\[
\varphi_i(v,g)
= f_i(v,g) + \frac{1}{n}\left( v(N) - \sum_{k \in N} f_k(v,g) \right).
\]
To reflect the informational role of $f$ in redistributing the surplus, we modify (FDS) accordingly.

\begin{itemize}
\item[] \textbf{$f$-fair distribution of the surplus ($f$-FDS):} 
For all $(v,g)$ and all $C,C' \in N/g$,
\[
\frac{1}{|C|}\!\left( \sum_{i \in C} (\varphi_i(v,g) - f_i(v,g)) \right)
= \frac{1}{|C'|}\!\left( \sum_{i \in C'} (\varphi_i(v,g) - f_i(v,g)) \right).
\]
\end{itemize}

($f$-FDS) requires that the surplus, measured relative to the reference solution $f$, is distributed equally across components—analogous to (FDS), but benchmarked against $f$ instead of component values.

Once (CE) is dropped, there are many solutions that satisfy (FA).
For example, consider a solution $f^\lambda=\lambda My+\phi$, where $\lambda \neq 1$ and $\phi: \mathcal{V}^N \rightarrow \mathbb{R}^n$ is an arbitrary solution except for the null solution (e.g., $\phi(v)=0$ for any $v \in \mathcal{V}^N$) on standard TU-games.
$f^{\lambda}$ satisfies (FA) since $\phi$ is independent of the communication network.

The next theorem extends the characterization of $EEMy$ in \citet{vandenbrinketal2012}, which can accommodate any such solution as a benchmark solution.
When the communication network is complete, the communication game reduces to the underlying TU-game. 
Consequently, the characterized solution in Theorem \ref{f-ESS network 1} reduces to that of Theorem \ref{f-ess/pd exogenous}, where different axioms dealing with general communication networks are used.
Hence, we also regard Theorem  \ref{f-ESS network 1} as a communication-game variant of Theorem~\ref{f-ess/pd exogenous}.

\begin{theorem}\label{f-ESS network 1}
Suppose that $f \in \mathcal{F}_{\mathscr{G}}$ satisfies (FA).
Then a solution $\varphi$ satisfies (E), (FA), and ($f$-FDS) if and only if 
\[
\varphi_i(v,g)
= f_i(v,g) + \frac{1}{n}\left( v(N) - \sum_{k \in N} f_k(v,g) \right)
\]
for all $(v,g)$ and all $i \in N$.
\end{theorem}

\subsection{Efficient-fair extension operators}\label{subsec EFOp}

As discussed in Subsection~\ref{subsec f-ESS communication}, fairness plays a key role in identifying desirable solutions for communication games.
We now define \textit{efficient-fair extension operators}, which extend Definition~\ref{def eeop} while preserving local fairness.

\begin{definition}\label{def efop}
Let $\mathcal{D} \subseteq \mathcal{F}_{\mathscr{G}}$.
An operator $\Phi: \mathcal{D} \to \mathcal{F}_{\mathscr{G}}$ is an \textit{efficient-fair extension operator} in $\mathcal{D}$ if:
\begin{itemize}
\item[(1)] For all $f \in \mathcal{D}$, $\Phi(f)$ satisfies (E);
\item[(2)] For all $f \in \mathcal{D}$ and all $v \in \mathcal{V}^N$, if $f$ satisfies (FA-$v$) at $v$, then $\Phi(f)$ also satisfies (FA-$v$) at $v$.
\end{itemize}
\end{definition}

The first condition coincides with Definition~\ref{def eeop}.
The second, which we call the \textit{local fairness preservation property}, requires that the extended solution preserves fairness at every $v$ where the original solution $f$ is locally fair.
That is, if $f$ satisfies fairness for certain games but not others, $\Phi(f)$ should retain fairness wherever it already holds.

The following axiom is the communication-game counterpart of the one introduced earlier.

\begin{itemize}
%\item[] \textbf{Equal treatment (ET):} 
%For any $f \in \mathcal{F}_{\mathscr{G}}$, $g \in \mathscr{G}^N$, and $i,j \in N$, if $f_i(\cdot,g)=f_j(\cdot,g)$, then $\Phi_i(f)(\cdot,g)=\Phi_j(f)(\cdot,g)$.

\item[] \textbf{Weak equality for equal surplus (WEES):} 
For any $f,f' \in \mathcal{F}_{\mathscr{G}}$, $g \in \mathscr{G}^N$, and $i \in N$, if
\[
\sum_{k \in N} f_k(\cdot,g) = \sum_{k \in N} f'_k(\cdot,g)
\quad\text{and}\quad
f_i(\cdot,g) = f'_i(\cdot,g),
\]
then $\Phi_i(f)(\cdot,g) = \Phi_i(f')(\cdot,g)$.
\end{itemize}

The next result characterizes the ESS operator as the unique efficient-fair extension operator for communication games.

\begin{theorem}\label{efficient-fair main}
An efficient-fair extension operator $\Phi: \mathcal{F}_{\mathscr{G}} \to \mathcal{F}_{\mathscr{G}}$ satisfies (WEES) and ensures that $\Phi(f)$ satisfies ($f$-FDS) for all $f \in \mathcal{F}_{\mathscr{G}}$ if and only if
\[
\Phi_i(f)(v,g)
= f_i(v,g) + \frac{1}{n}\left( v(N) - \sum_{k \in N} f_k(v,g) \right)
\]
for all $(v,g)$ and all $i \in N$.
\end{theorem}

Independence of the axioms of Theorem \ref{efficient-fair main} is shown in Appendix \ref{appendix: independence_network}.
Although the axioms resemble those in Theorem~\ref{operator main result}, the proof strategy differs. 
Here, we do not rely on the technique of \citet{Casajus2015}. 
Instead, we use the characterization in Theorem~\ref{f-ESS network 1}, which naturally fits the local fairness preservation property.
The detailed proof is provided in Appendix~\ref{Appendix_proof_communication}.

\section{Efficient extension for solutions of coalition structure games}\label{sec_coalition}

In this section, we consider extension operators in a different enriched model of TU-games, called games with coalition structures \citep[e.g.,][]{aumanndreze1974, owen1977}, or coalition structure games for short.
The analysis in this section runs in exact parallel to Section~\ref{sec communication game}, 
thereby illustrating that our framework for efficient extension operators applies 
beyond communication networks to coalition structure games.
We show the parallel results in Section \ref{sec communication game} by introducing a suitable counterpart of the fairness property in this framework.

\subsection{$f$-ESS value for coalition structure games}\label{subsec_coalition}

In Section \ref{sec communication game}, we preserve the fairness criterion (FA) proposed by \cite{myerson1977graphs} to consider a new rule by modifying the underlying rules in Theorem \ref{f-ESS network 1} and its generalization to efficient-fair extension operators in Theorem \ref{efficient-fair main}.

Within the framework of coalition structure games, \citet{slikker2000} introduces the concept of Restricted Balanced Contributions (RBC) as an analogue of (FA), and characterizes the Aumann--Dr\`eze value—combined with component efficiency—in a manner analogous to the characterization of the Myerson value.
As we formally discuss below, since the Aumann--Dr\`eze value can be viewed as a special case of the Myerson value for a network in which each component is fully connected, (RBC) is indeed regarded as the counterpart of (FA) in the context of coalition structure games.
However, we will show that (RBC) is too strong a requirement when considering ESS-type rules within this setting.
To address this technical difficulty, we introduce a weaker version of (RBC), based on the concept of Balanced Cycle Contributions (BCC) proposed by \citet{kamijokongo2010}.
Since (BCC) is applicable only within a variable player set framework, we accordingly extend our model in what follows.

Let $\mathcal{U}$ be a universal countably infinite player set, and let each player set $N \subset \mathcal{U}$ be finite.  
For $N$, let $\mathcal{V}^N$ denote all TU-games $(N,v)$. 
For $S \subseteq N$, $(S,v|_S) \in \mathcal{V}^S$ is the subgame of $(N,v)$.  
A coalition structure is a partition $\mathcal{P}=\{C_1,\dots,C_m\}$ of $N$, and we denote the set of all partitions by $\Pi^N$.   
A coalition structure game is a triplet $(N,v,\mathcal{P}) \in \mathcal{V}^N \times \Pi^N$, and we denote the set of all such games by $\mathcal{CV} = \bigcup_{N\subset \mathcal{U}, |N|<\infty} (\mathcal{V}^N \times \Pi^N)$.  
A solution is a mapping $\varphi: \mathcal{CV} \to \bigcup_{N\subset \mathcal{U},|N|<\infty} \mathbb{R}^{|N|}$, and let $\mathcal{F}_{\mathcal{C}}$ denote the set of all solutions.

For the counterpart of the Myerson value, the \textit{Aumann-Dr\`eze value} \citep{aumanndreze1974} is defined as follows: for any $(N, v, \mathcal{P}) \in \mathcal{CV}$ and $i \in C \in \mathcal{P}$,
\[
AD_i(N, v, \mathcal{P})=Sh_i(C, v|_C),
\]
where $v|_C \in \mathcal{V}^C$ is the sub-game of $v$ for the player set $C$ (i.e., $v|_C(S)=v(S)$ for any $S \subseteq C$).
Indeed, if every player in a component $C \in N/g$ is connected, that is $g|_C$ is a complete network, $AD_i(N, v, \mathcal{P})=My_i(N, v,g)$ for any $i \in C$.
Accordingly, the equal surplus sharing rule for the given $f \in \mathcal{F}_{\mathcal{C}}$ is defined as follows: for any $(N, v, \mathcal{P}) \in \mathcal{CV}$ and $i \in C \in \mathcal{P}$,
\[
\varphi_i(N, v, \mathcal{P})=f_i(N, v, \mathcal{P})+\frac{1}{n}\Bigl(v(N)-\sum_{k \in N}f_k(N, v, \mathcal{P})   \Bigr).
\]

The corresponding axiom of efficiency and ($f$-FDS) can be written as follows.
\begin{itemize}
\item[] \textbf{Component efficiency (CE)}: For any $(N, v, \mathcal{P}) \in \mathcal{CV}$ and $C \in \mathcal{P}$, $\sum_{i \in C}\varphi_i(N, v,\mathcal{P})=v(C)$.

\item[] \textbf{Efficiency (E)}: For any $(N, v, \mathcal{P}) \in \mathcal{CV}$, $\sum_{i \in N}\varphi_i(N, v,\mathcal{P})=v(N)$.

\item[] \textbf{$f$-fair distribution of the surplus for coalition structures ($f$-FDSC)}:  For any $(N, v, \mathcal{P}) \in \mathcal{CV}$ and $C, C' \in \mathcal{P}$,
\[
\frac{1}{|C|}\Bigl(\sum_{i \in C}\varphi_i(N, v, \mathcal{P})-f_i(N, v, \mathcal{P}) \Bigr)=\frac{1}{|C'|}\Bigl(\sum_{i \in C'}\varphi_i(N, v, \mathcal{P})-f_i(N, v, \mathcal{P}) \Bigr).
\]
\end{itemize}

\cite{slikker2000} shows that the Aumann-Dr\`eze value is the unique solution that satisfies (CE) and (CRBC).
\begin{itemize}
\item[] \textbf{Component restricted balanced contributions (CRBC)}: For any $(N, v, \mathcal{P}) \in \mathcal{CV}$ and $i,j \in C \in \mathcal{P}$, $\varphi_i(N, v, \mathcal{P})-\varphi_i(N, v, \mathcal{P}_{-j})=\varphi_j(N, v, \mathcal{P})-\varphi_j(N, v, \mathcal{P}_{-i})$, where $\mathcal{P}_{-l}=\bigl( \mathcal{P}\setminus \{C\} \bigr) \cup \{C\setminus\{l\}, \{l\}\}, l=i,j$.
\end{itemize}

Component restricted balanced contributions (CRBC) requires that, for each pair of $i,j$ in the same coalition, the effect on the payoff difference by the deletion of the opposite player is symmetric.
At first glance, (CRBC) is a counterpart of (FA) in this framework.
However, it is a much stronger requirement than (FA).
To see this, for a component $C \in N/g$, suppose $g|_C$ is a complete network.
Then, for any $i,j \in C$, the requirement of $\varphi_i(N, v, \mathcal{P})-\varphi_i(N, v, \mathcal{P}_{-j})=\varphi_j(N, v, \mathcal{P})-\varphi_j(N, v, \mathcal{P}_{-i})$ corresponds to
\begin{equation}
\varphi_i(v, g)-\varphi_i(v, g-N_j(g))=\varphi_j(v, g)-\varphi_j(v, g-N_i(g)), \label{eq: FA}
\end{equation}
where $g-N_i(g)=g\setminus\bigcup_{j \in N_i(g)} \{ij\}$ with abuse of notation.
If $C=\{i,j\}$, this equation holds by (FA) because $g-N_i(g)=g-N_j(g)=g-ij$.
However, (FA) does not guarantee  \eqref{eq: FA} in general.

An equal surplus sharing rule for $f \in \mathcal{F}_{\mathcal{C}}$ does not satisfy (CRBC), even if $f$ satisfies (CRBC) in general.
For example, the efficient extension of the Aumann--Dr\`eze value does not satisfy (CRBC), which stands in sharp contrast to the fact that both the Myerson value and the efficient egalitarian  Myerson value satisfy (FA). This contrast highlights the importance of designing an operator that restores efficiency while preserving fairness-type properties across different environments.

This observation indicates the necessity of introducing a weaker axiom in order to accommodate an ESS-type extension of the solution within this framework. 
To this end, we adopt the following weakened version of the balanced cycle contribution property, originally introduced by \cite{kamijokongo2010}.\footnote{
As discussed in \cite{kamijokongo2010}, the requirement of considering cycles of any length can be weakened to considering only cycles of length three. 
It is also noted that (BCC) is a substantially weaker condition than (BC), as it is satisfied by many solutions. 
For further details on this property, see \cite{kamijokongo2010}.}

\begin{itemize}
\item[] \textbf{Restricted balanced cycle contributions (RBCC-$v$)}: Fix $(N, v) \in \mathcal{V}^N$.
For any $C \in \mathcal{P} \in \Pi^N$,
\[
\sum_{l=1}^{|C|} \Bigl( \varphi_{i_l}(N, v, \mathcal{P})-\varphi_{i_l}(N\setminus \{i_{l-1}\}, v, \mathcal{P}_{-i_{l-1}}) \Bigr)=\sum_{l=1}^{|C|} \Bigl( \varphi_{i_l}(N, v,  \mathcal{P})-\varphi_{i_l}(N\setminus \{i_{l+1}\}, v, \mathcal{P}_{-i_{l+1}}) \Bigr),
\]
where $C$ is enumerated as $C=\{i_1, i_2, \ldots, i_{|C|}\}$ with $i_0=i_{|C|}$ and $i_{|C|+1}=i_1$.

\item[] \textbf{Restricted balanced cycle contributions (RBCC)}: $\varphi$ satisfies (RBCC-$v$) for any $(N,v) \in \mathcal{V}^N$ such that $N \subset \mathcal{U}$ with $|N|<\infty$.
 
\end{itemize}

Restricted balanced cycle contributions (RBCC) only requires that the sum of the payoff difference from the deletion of players is invariant for the order in each component.
Notice that (RBCC) cannot work in the finite player framework because it is silent about the relationship among the values on, for example, $\mathcal{P}$, $\mathcal{P}_{-i}$ and $\mathcal{P}_{-j}$, which is specified in (CRBC).
Technically speaking, this difference makes it difficult to use induction on the size of components, which is a key logic in \cite{slikker2000} and Theorem \ref{f-ESS network 1}.
Note also that it requires nothing for a two-player component.
For such components, we consider the following property.

\begin{itemize}
\item[] \textbf{$f$-Equal gain relative to null players ($f$-EGN)}:  For any $(N, v, \mathcal{P}) \in \mathcal{CV}$ and $i, j \in C \in \mathcal{P}$, if $j$ is null in $v$, then $\varphi_i(N, v, \mathcal{P})-\varphi_j(N,v, \mathcal{P})=f_i(N,v,\mathcal{P})-f_j(N,v,\mathcal{P})$.
\end{itemize}

The following result provides a counterpart of Theorem \ref{f-ess/pd exogenous} parallel to Theorem \ref{f-ESS network 1} in communication games.

\begin{theorem}\label{f-ESS_coalition}
Suppose that a solution $f \in \mathcal{F}_{\mathcal{C}}$ satisfies (RBCC).
Then, a solution $\varphi$ satisfies (E), (RBCC), ($f$-EGN), and ($f$-FDSC) if and only if 
\[
\varphi_i(N, v, \mathcal{P})=f_i(N, v, \mathcal{P})+\frac{1}{n}\Bigl(v(N)-\sum_{k \in N}f_k(N, v, \mathcal{P})   \Bigr)
\]
 for any $(N, v, \mathcal{P}) \in \mathcal{CV}$ and $i \in N$.
\end{theorem}

If $f=AD$, this defines the \textit{efficient egalitarian Aumann--Dr\`eze value}.
To the best of our knowledge, the formal definition of the solution and its axiomatic foundation have not been investigated in the literature, although efficient extensions of solutions in communication games are widely studied, as we explained in Section \ref{sec communication game}.
Hence, Theorem \ref{f-ESS_coalition} itself would be of independent interest.

To illustrate how (RBCC) and ($f$-EGN) work for the proof, take any $(N,v, \mathcal{P}) \in \mathcal{CV}$ and let us assume that $C=\{i,j\} \in \mathcal{P}$ is a two-player component.
Recall that (RBCC) does not require anything to $\varphi_i$ and $\varphi_j$ in this case.
Choose $k \in \mathcal{U}\setminus N$ and let $N'=N \cup \{k\}$.
Consider a modified game $(N', v', \mathcal{P}')$ and $\mathcal{P'}=(\mathcal{P}\setminus \{C\}) \cup \{C\cup \{k\}\}$.
We also assume that $k$ is a null player in $v'$ and $v'|_N=v$.
\black
Then, by (RBCC), we obtain the equation
\begin{align*}
&\varphi_i(N'\setminus\{k\},v',\mathcal{P'}_{-k})+\varphi_j(N'\setminus\{i\},v',\mathcal{P'}_{-i})+\varphi_k(N'\setminus\{j\},v',\mathcal{P'}_{-j})\\
&=\varphi_i(N'\setminus\{j\},v',\mathcal{P'}_{-j})+\varphi_j(N'\setminus\{k\},v',\mathcal{P'}_{-k})+\varphi_k(N'\setminus\{i\},v',\mathcal{P'}_{-i})\\
\Leftrightarrow~~& \varphi_i(N,v,\mathcal{P})-\varphi_j(N,v,\mathcal{P})\\
&=\Bigl(\varphi_i(N'\setminus\{j\},v',\mathcal{P'}_{-j})-\varphi_k(N'\setminus\{j\},v',\mathcal{P'}_{-j})\Bigr)
-\Bigl(\varphi_j(N'\setminus\{i\},v',\mathcal{P'}_{-i})-\varphi_k(N'\setminus\{i\},v',\mathcal{P'}_{-i})\Bigr).
\end{align*}
Then, by ($f$-EGN), the right-hand side of the above equation is equivalent to
\[
\Bigl(f_i(N'\setminus\{j\},v',\mathcal{P'}_{-j})-f_k(N'\setminus\{j\},v',\mathcal{P'}_{-j})\Bigr)
-\Bigl(f_j(N'\setminus\{i\},v',\mathcal{P'}_{-i})-f_k(N'\setminus\{i\},v',\mathcal{P'}_{-i})\Bigr),
\]
which is an already determined constant because $f$ is fixed as the underlying solution.
We can also show that, by (E) and ($f$-FDSC), 
$\varphi_i(N,v,\mathcal{P})+\varphi_j(N,v,\mathcal{P})$ is constant (see Lemma \ref{E+FDS} for the corresponding part for communication games).
Therefore, we obtain the solvable equations for $\varphi_i(N,v,\mathcal{P})$ and $\varphi_j(N,v,\mathcal{P})$, which implies the uniqueness of $\varphi$.

Note that, although the proof for the uniqueness is similar for Theorem \ref{f-ESS network 1}, the way to derive the solvable equations is different: While (FA) directly induces the system in Theorem \ref{f-ESS network 1}, we need an additional player outside of the underlying player set $N$ as a null player to use (RBCC) and ($f$-EGN) for Theorem \ref{f-ESS_coalition}.

\subsection{Efficient-RBCC extension operators}
We consider corresponding efficient extension operators by adjusting Definition \ref{def efop} to this setup in the similar manner as Section \ref{sec communication game}.

\begin{definition}\label{def efop_coalition}
Let $\mathcal{D} \subseteq \mathcal{F}_{\mathcal{C}}$ be the set of solutions.
An operator $\Phi: \mathcal{D} \rightarrow  \mathcal{F}_{\mathcal{C}}$ is an efficient-RBCC extension operator in $\mathcal{D}$ if
\begin{itemize}
\item[(1)] for any $f \in \mathcal{D}$, $\Phi(f)$ satisfies (E),

\item[(2)] for any $f \in \mathcal{D}$, if $f$ satisfies (RBCC-$v$) at $v$, then $\Phi(f)$ also satisfies (RBCC-$v$) at $v$.
\end{itemize} 
\end{definition}

The first property is the same as in the previous sections.
The second property shares the same spirit as the local fairness preservation for efficient extension operators in communication games.
As we discussed in Subsection \ref{subsec_coalition}, (CRBC) is a strong requirement that ESS operator does not satisfy in general.
Therefore, we require local preservation of (RBCC) in this case.

The following axiom for an operator $\Phi$ is a straightforward generalization of the corresponding one in previous sections.

\begin{itemize}
%\item[] \textbf{Equal treatment (ET)}: 
%For any $f \in \mathcal{D}$, $\mathcal{P} \in \Pi^N$, and $i, j \in N$, if $f_i(N, \cdot,\mathcal{P})=f_j(N, \cdot,\mathcal{P})$, then $\Phi_i(f)(N, \cdot, \mathcal{P})=\Phi_j(f)(N, \cdot, \mathcal{P})$.  

\item[] \textbf{Weak equality for equal surplus (WEES)}: 
For any $f, f' \in \mathcal{D}$, $\mathcal{P} \in \Pi^N$, and $i \in N$, if 
\[
\sum_{k \in N}f_k(N, \cdot,\mathcal{P})= \sum_{k \in N}f'_k(N, \cdot,\mathcal{P})~\text{and}~f_i(N, \cdot,\mathcal{P})=f'_i(N, \cdot,\mathcal{P}),
\]
then $\Phi_i(f)(N, \cdot,\mathcal{P})=\Phi_i(f')(N, \cdot,\mathcal{P})$.
\end{itemize}

The following result characterizes the ESS operator for the solution of coalition structure games.
The proof is similar to Theorem \ref{efficient-fair main} and it is relegated to Appendix \ref{Appendix_proof_coalition}.

\begin{theorem}\label{ESS coalition}
An efficient-RBCC extension operator $\Phi: \mathcal{F}_{\mathcal{C}} \rightarrow \mathcal{F}_{\mathcal{C}}$ satisfies (WEES) and $\Phi(f) \in \mathcal{F}_{\mathcal{C}}$ satisfies  ($f$-EGN) and ($f$-FDSC)  for any $f \in \mathcal{F}_{\mathcal{C}}$ if and only if it is the ESS operator.
\end{theorem}

Independence of the axioms of Theorem \ref{ESS coalition} is shown in Appendix \ref{appendix: independence_coalition}.
Together with the results for communication games, Theorem \ref{efficient-fair main}, this theorem confirms that our operator-based approach unifies efficiency-restoring extensions across distinct enriched models of TU-games.

We remark that (RBCC) is not the only possible way to incorporate the ESS operator within this framework. Identifying other, possibly weaker, versions of (RBC) that could lead to the same result is beyond the scope of this paper, but remains an interesting direction for future research. See also the discussion of related issues in Subsection \ref{subsec:preserve}.

\section{Discussions}\label{sec_conclusingremarks}

In this paper, we have introduced a novel framework for justifying efficient solutions in TU-games, focusing on cases in which such solutions are obtained by modifying initially inefficient ones.
Our approach differs from the standard axiomatic method that characterizes a single solution via a fixed set of axioms. 
Instead, we study operators that transform one solution into another, thereby providing a new perspective on the structure and interrelation of solution concepts.

We discuss several promising directions for future research.
These include extensions to other classes of operators, examining the compatibility of efficiency with alternative fairness axioms, extending our approach to set-valued solutions, and exploring applications to broader classes of resource-allocation problems beyond TU-games.

\subsection{Other classes of operators}\label{sec_otherop}

Throughout the paper, we have focused on the characterization of the ESS and PS operators, as these two methods are most commonly employed.
Theorem \ref{operator main result} shows that these two methods share a common mathematical structure—summarized by (ET) and (EES)—and differ only in the situational axiom.
A natural next step is to explore how other extension operators can be characterized by altering these situational axioms.
While pursuing this direction is beyond the scope of the present paper, we briefly discuss possible extensions.

Recall that the PS value can also be written as
\[
\varphi_i(v)=f_i(v)+\frac{f_i(v)}{\sum_{k \in N}f_k(v)}\left(v(N)-\sum_{k \in N}f_k(v)\right).
\]
Accordingly, we may define a more general class of operators that encompasses both the ESS and PS operators: for any $f \in \mathcal{F}$, $i \in N$, and $v \in \mathcal{V}^N$,
\[
\Phi_i(f)(v)=f_i(v)+w_i(f(v))\left(v(N)-\sum_{k \in N}f_k(v)\right),
\]
where $w_i(f(v))=w_j(f(v))$ whenever $f_i(v)=f_j(v)$ and $\sum_{k \in N}w_k(f(v))=1$.
This operator also satisfies both (ET) and (EES).
A special case is a convex combination of the ESS and PS operators, given by $w_i(f(v))=\alpha/n+(1-\alpha)\frac{f_i(v)}{\sum_{k \in N}f_k(v)}$.
Convex combinations of solutions in TU-games have been widely studied \citep[e.g.,][]{vandenBrinketal2013, CasajusHuettner2014JET, yokotefunaki2017, Nakada2024}.
Hence, the axioms and techniques developed in those studies would be useful for extending our characterization results in this direction.

We may also consider an alternative approach that relaxes (ET) to accommodate asymmetry among players.
For instance, the condition $w_i(f(v))=w_j(f(v))$ may not hold even when $f_i(v)=f_j(v)$.
In line with studies on convex combinations of solutions in TU-games, \citet{abe2019weighted} provide a theoretical foundation for incorporating such asymmetries.
Their technique could prove valuable in pursuing such generalizations.

\subsection{Preservation of axioms}\label{subsec:preserve}

As discussed in Section \ref{sec communication game}, the concept of an efficient extension operator provides a useful framework for justifying common efficiency-adjustment or normalization methods in solution concepts, even in the presence of communication structures that restrict feasible coalitions.
In particular, following Myerson’s original fairness axiom, we examined extension operators that satisfy both efficiency and fairness—referred to as efficient-fair extension operators.
However, there can be trade-offs between efficiency and other fairness-like axioms.
\cite{beal2012fairness} propose alternative notions of fairness in communication games and provide axiomatic foundations for corresponding solution concepts.
Building on these results, \cite{bealetal2018} investigate whether efficient extensions can respect such alternative fairness axioms.
Their findings reveal that, in certain cases, some fairness-like axioms are incompatible with efficiency, implying that an extension operator preserving both properties may not exist.
In light of these impossibility results, identifying which axioms can be consistently preserved through efficient extension remains an important direction for future research.

\subsection{Set-valued solutions}

In this paper, we have restricted our attention to single-valued solutions and their efficient extensions.
A natural direction for future work is to extend the framework to set-valued solutions.
In particular, one may begin with either a single-valued or a set-valued solution and seek an efficient extension of it in the set-valued sense.
For instance, by restricting attention to balanced games \citep{shapley1971cores}, one could examine how a given solution concept extends to the Core. 
%Such an approach could yield new insights into the axiomatic foundations and interpretations of the Core itself.

\subsection{Beyond TU-games}

From a technical standpoint, our main results build on the methodological framework developed by \citet{Casajus2015}, who characterized simple lump-sum tax redistribution rules via efficiency, monotonicity, and symmetry.
His model is mathematically equivalent to additive TU-games, where the resulting solution corresponds to the egalitarian Shapley value.
In this sense, his setting can be interpreted as a special application within TU-game theory.
At the same time, his approach is sufficiently flexible to encompass a wide range of resource-allocation models where individual payoffs and aggregate resources are explicitly specified.

Specifically, the monotonicity axiom introduced by \citet{Casajus2015} depends only on individual incomes and their aggregate total. 
Consequently, the resulting allocation rule can be expressed as a function of these two parameters.
Subsequent studies have explored taxation rules for non-negative incomes \citep{casajus2015monotonic, casajus2016differentially, yokote2017weak, martinez2022laissez, zou2023axiomatizations}, as well as asymmetric cases \citep{abe2017monotonic}.\footnote{The same model originates from \cite{ju2007non}, who analyze redistribution rules based on non-manipulability within subgroups.}
These generalizations are still developed within the additive TU-game framework, and thus the broader technical potential of \citet{Casajus2015}’s approach has yet to be fully explored.\footnote{A few exceptions include \cite{chambers2021bilateral} and \cite{martinez2024redistribution}, both of which extend redistribution problems by introducing heterogeneity in basic needs.}

Our proofs for the characterizations of the ESS and PS operators leverage the technique of \citet{Casajus2015}, encapsulated in the condition (EES), within a more general framework that transforms one allocation rule into another.
Therefore, while our analysis has been situated within TU-games and has focused on efficient extensions of their solution concepts, the key logic underlying our argument is applicable to broader environments beyond TU-games.
Exploring efficient extensions in such general settings would be a fruitful avenue for future research.

\section{Concluding remarks}\label{sec_conclude}

This paper develops a general framework for transforming inefficient solutions in TU-games into efficient ones through the notion of an efficient extension operator.
By shifting the focus from the characterization of individual solutions to the characterization of mappings between solutions, the paper establishes a new methodological foundation for analyzing how efficiency interacts with existing normative principles in economic allocation problems.
This operator-based approach unifies and generalizes well-known efficiency-adjustment methods such as the egalitarian and proportional extensions, providing new axiomatic justifications for a broad class of efficiency-adjusted values.
Beyond these specific applications, the framework clarifies the logical connections among solution concepts and offers a versatile tool for constructing and evaluating new ones.

While the analysis has centered on TU-games, the underlying logic—transforming a baseline rule into an efficient one through universally defined axioms—extends naturally to more general allocation environments, including redistribution, taxation, and mechanism design problems.
This broader applicability highlights the potential of the operator-based approach as a unifying language for reasoning about efficiency-adjustment mechanisms across different domains of economic theory.

Future research may further explore how the axiomatic properties of operators interact with other normative or behavioral principles, thereby enriching the study of efficiency adjustments beyond the cooperative game framework.
We hope that this perspective will serve as a foundation for further theoretical developments and applications in cooperative game theory and related areas of economic design.
\\

\textbf{Declaration of conflict of interest}: None.

\begin{center}
\Large{{\bf Appendix}}
\end{center}

\numberwithin{definition}{section}
\numberwithin{theorem}{section}
\numberwithin{lemma}{section}
\numberwithin{proposition}{section}
\numberwithin{corollary}{section}
\numberwithin{example}{section}
\renewcommand{\theequation}{\thesection.\arabic{equation}}
\appendix

\section{Proofs in Section \ref{sec_main}}\label{appendix_proof_main}

\begin{proof}[Proof of Theorem \ref{operator main result}]
(1) The if part is trivial. We show the only if part.
Let $\Phi$ be an efficient extension operator that satisfies (ET), (EES), and suppose that $\Phi(f) \in \mathcal{F}$ satisfies ($f$-IES) for any $f \in \mathcal{F}$.

Take any $f \in \mathcal{F}$, $i \in N$, and $v \in \mathcal{V}^N$.
By ($f$-IES) for each $f$, we can write $\Phi: \mathcal{F}\rightarrow \mathcal{F}$ as $\Phi_i(f)(v)=\gamma^{f}_i\bigl( f_i(v), v(N)-\sum_{k \in N}f_k(v)\bigr)$, where $\gamma^f_i: \mathbb{R}^2\rightarrow \mathbb{R}$.
Note that the function $\gamma_i^f$ depends on $f$, but does not depend on $v$.
Moreover, by (EES), for any $f' \in \mathcal{F}$ and $v \in \mathcal{V}^N$ with $f_i(v)=f'_i(v)$ and $\sum_{k \in N}f_k(v)=\sum_{k \in N}f'_k(v)$, 
\begin{align*}
\gamma^f_i\bigl( f_i(v), v(N)-\sum_{k \in N}f_k(v) \bigr)&=\Phi_i(f)(v)\\
&=\Phi_i(f')(v)\\
&=\gamma^{f'}_i\bigl( f'_i(v), v(N)-\sum_{k \in N}f'_k(v) \bigr)\\
&=\gamma^{f'}_i\bigl( f_i(v), v(N)-\sum_{k \in N}f_k(v) \bigr),
\end{align*}
which implies that $\gamma^f_i=\gamma^{f'}_i$ for any $f, f' \in \mathcal{F}$.
Therefore, there exists  $\gamma_i:\mathbb{R}^2\rightarrow \mathbb{R}$ such that 
\begin{equation}
\Phi_i(f)(v)\equiv \gamma_i\bigl( f_i(v), v(N)-\sum_{k \in N}f_k(v) \bigr),  \forall f \in \mathcal{F}, \forall v \in \mathcal{V}^N. \label{proof_function}
\end{equation}

Now, fix arbitrary $f \in \mathcal{F}$, $i \in N$ and $v \in \mathcal{V}^N$, and let us consider the following solution $f' \in \mathcal{F}$:
\[
f'_j(v')=f_i(v), \forall j \in N, \forall v' \in \mathcal{V}^N.
\]
By construction, $f'_i=f'_j$ for any $i,j \in N$, so that $\Phi_i(f')=\Phi_j(f')$ for any $i,j \in N$ by (ET).
Since $\Phi(f') \in \mathcal{F}$ satisfies (E), 
\begin{equation}
\Phi_i(f')(v')=\frac{v'(N)}{n},\forall i \in N, \forall v' \in \mathcal{V}^N. \label{proof_ET}
\end{equation}

Next, let us consider the game $\tilde v \in \mathcal{V}^N$ such that $\tilde v(N)=nf_i(v)+\bigl( v(N)-\sum_{k \in N}f_k(v)\bigr)$.
Then, by construction, we have $f'_i(\tilde v)=f_i(v)$ and $\tilde v(N)-\sum_{k \in N}f'_k(\tilde v)=\tilde v(N)-nf_i(v)= v(N)-\sum_{k \in N}f_k(v)$.
Therefore, by \eqref{proof_function} and \eqref{proof_ET}, we have
\begin{align*}
\Phi_i(f)(v)&=\gamma_i\Bigl( f_i(v), v(N)-\sum_{k \in N}f_k(v)\Bigr)\\
&=\gamma_i\Bigl( f'_i(\tilde v), \tilde v(N)-\sum_{k \in N}f'_k(\tilde v)\Bigr)\\
&=\Phi_i(f')(\tilde v)\\
&=\frac{\tilde v(N)}{n}\\
&= f_i(v)+\frac{1}{n} \left( v(N)-\sum_{k \in N}f_k(v) \right).
\end{align*}
Since $f \in \mathcal{F}$, $i \in N$, and $v \in \mathcal{V}^N$ are arbitrary chosen, $\Phi$ is the ESS operator.

(2) Take any $f \in \mathcal{F}_+$, $i \in N$, and $v \in \mathcal{V}_+^N$.
The proof is similar to that of (1), but we need some modifications. First, modify the domains to $\mathcal{F}_+$ and $\mathcal{V}^N_+$. Second, modify \eqref{proof_function} using (EES) and ($f$-IER) as follows: 
\begin{equation}
\Phi_i(f)(v)=\hat \Phi_i\Bigr(f_i, \sum_{k \in N}f_i \Bigr)(v)=\gamma_i \Bigl(f_i(v), \frac{v(N)}{\sum_{k \in N}f_k(v)} \Bigr), \forall v \in \mathcal{V}^N_+. \label{proof_function2}
\end{equation}
We can also have \eqref{proof_ET} in the modified domains by the same discussions for arbitrary fixed $f \in \mathcal{F}_+$ and $v \in \mathcal{V}^N_+$.
Let us consider the game $\hat v$ such that $\hat v(N)=nf_i(v)v(N)/\sum_{k \in N}f_k(v)$.
Then, by construction, we have $f'_i(\hat v)=f_i(v)$ and $\hat v(N)/\sum_{k \in N}f'_k(\hat v)=\hat v(N)/(nf_i(v))=v(N)/\sum_{k \in N}f_k(v)$.
Therefore, by \eqref{proof_ET} and \eqref{proof_function2}, we have
\begin{align*}
\Phi_i(f)(v)&=\gamma_i\Bigl( f_i(v), \frac{v(N)}{\sum_{k \in N}f_k(v)}\Bigr)\\
&=\gamma_i\Bigl( f'_i(\hat v), \frac{\hat v(N)}{\sum_{k \in N}f'_k(\hat v)}\Bigr)\\
&= \Phi_i(f')(\hat v)\\
&=\frac{\hat v(N)}{n}\\
&=\frac{f_i(v)}{\sum_{k \in N}f_k(v)}v(N).
\end{align*}
Since $f \in \mathcal{F}_+$, $i \in N$, and $v \in \mathcal{V}_+^N$ are arbitrary chosen, $\Phi$ is the PS operator.
\end{proof}

\section{Cohesively efficient extension} \label{Appendix_cohesiveefficiency}

\citet{beal2021cohesive} introduced the concept of \textit{cohesive efficiency}, which requires that a solution achieve the maximum total worth that may arise from smaller coalitions acting independently. 
Formally, the concept is defined as follows.

\begin{itemize}
    \item[] \textbf{Cohesive efficiency (CoE)}: For any $v \in \mathcal{V}^N$, $\sum_{i \in N}\varphi_i(v)=\max_{\mathcal{P} \in \Pi^N}\sum_{T \in \mathcal{P}}v(T)$.
\end{itemize}
Note that (CoE) and (E) coincide at $v$ if and only if $\mathcal{P}=\{N\}$ is the maximizer of the right-hand side.
We define \textit{cohesively efficient extension operator} as follows.

\begin{definition}
Let $\mathcal{D} \subseteq \mathcal{F}$ be the set of solutions.
An operator $\Phi: \mathcal{D}\rightarrow \mathcal{F}$ is a cohesively efficient extension operator in $\mathcal{D}$ if, for any $f \in \mathcal{D}$, $\Phi(f)$ satisfies (CoE).
\end{definition}
To accommodate (CoE) based on Theorem \ref{operator main result}, we consider the following modification of situational axioms.

\begin{itemize}
\item[] \textbf{$f$-individualistic property for equal cohesive surplus ($f$-IECoS)}: 
For any $v,w \in  \mathcal{V}^N$ and $i \in N$, if 
\[
\max_{\mathcal{P} \in \Pi^N}\sum_{T \in \mathcal{P}}v(T)-\sum_{k \in N}f_k(v)=\max_{\mathcal{P} \in \Pi^N}\sum_{T \in \mathcal{P}}w(T)-\sum_{k \in N}f_k(w)~~\text{and}~~f_i(v)=f_i(w),
\]
then $\varphi_i(v)=\varphi_i(w)$.

\item[] \textbf{$f$-individualistic property for equal cohesive ratio ($f$-IECoR)}: 
For any $v,w \in  \mathcal{V}^N_{+}$ and $i \in N$, if 
\[
\frac{\max_{\mathcal{P} \in \Pi^N}\sum_{T \in \mathcal{P}}v(T)}{\sum_{k \in N}f_k(v)}=\frac{\max_{\mathcal{P} \in \Pi^N}\sum_{T \in \mathcal{P}}w(T)}{\sum_{k \in N}f_k(w)}~~\text{and}~~f_i(v)=f_i(w),
\]
then $\varphi_i(v)=\varphi_i(w)$.

\end{itemize}

By the same logic for the proof of Theorem \ref{operator main result}, we can obtain the following characterizations of cohesively efficient extension operators.

\begin{corollary}\label{cohesive_op}
\begin{itemize}
\item[] 
\item[(1)] An cohesively efficient extension operator $\Phi: \mathcal{F} \rightarrow \mathcal{F}$ satisfies (ET) and (EES), and $\Phi(f)$ satisfies ($f$-IECoS) for any $f \in \mathcal{F}$ if and only if
\[
\Phi_i(f)(v)=f_i(v)+\frac{1}{n} \left( \max_{\mathcal{P} \in \Pi^N}\sum_{T \in \mathcal{P}}v(T)-\sum_{k \in N}f_k(v) \right).
\]
for any $f \in \mathcal{F}$, $v \in \mathcal{V}^N$ and $i \in N$.

\item[(2)] An efficient extension operator $\Phi: \mathcal{F}_+ \rightarrow \mathcal{F}_+$ satisfies (ET) and (EES), and  $\Phi(f)$ satisfies ($f$-IECoR) for any $f \in \mathcal{F}_+$ if and only if 
\[
\Phi_i(f)(v)=\frac{f_i(v)}{\sum_{k \in N}f_k(v)}\Bigl(\max_{\mathcal{P} \in \Pi^N}\sum_{T \in \mathcal{P}}v(T)\Bigr).
\]
for any $f \in \mathcal{F}_+$, $v \in \mathcal{V}^N_{+}$ and $i \in N$,
\end{itemize}
\end{corollary}

As special cases of the expression defined in (1) of Corollary \ref{cohesive_op}, if $f=Sh$, the stand-alone value (i.e, $f_i(v)=v(\{i\})$), or the equal division value (i.e, $f_i(v)=v(N)/n$), then the corresponding solutions coincide with the cohesive Shapley value, the cohesive equal surplus division value, and the cohesive equal division value, respectively, analyzed by \citet{beal2021cohesive}.

\begin{proof}[Proof of Corollary \ref{cohesive_op}]
For each $v \in \mathcal{V}^N$ (resp. $\mathcal{V}^N_+$), by abuse of notation, we write $v(\mathcal{P}^*(v))\equiv\max_{\mathcal{P} \in \Pi^N}\sum_{T \in \mathcal{P}}v(T)$, where $\mathcal{P}^*(v)$ is a maximizer of the right-hand side.

(1) Take any $f \in \mathcal{F}$, $i \in N$, and $v \in \mathcal{V}^N$.
By the same arguments in the proof for (1) of Theorem \ref{operator main result} replacing ($f$-IES) with ($f$-IECoS), we have

\begin{equation}
\Phi_i(f)(v)\equiv \gamma_i\bigl( f_i(v), v(\mathcal{P}^*(v))-\sum_{k \in N}f_k(v) \bigr),  \forall f \in \mathcal{F}, \forall v \in \mathcal{V}^N. \label{proof_function_cohesive}
\end{equation}
and 
\begin{equation}
\Phi_i(f')(v')=\frac{v'(\mathcal{P}^*(v'))}{n},\forall i \in N, \forall v' \in \mathcal{V}^N, \label{proof_ET_cohesive}
\end{equation}
where $f'_j(v')=f_i(v)$ for any $j \in N$ and $v' \in \mathcal{V}^N$.
Then, let us define the game $\tilde v \in \mathcal{V}^N$ such that $\tilde v(\mathcal{P}^*(\tilde v))=nf_i(v)+\bigl( v(\mathcal{P}^*(v))-\sum_{k \in N}f_k(v) \bigr)$.
By \eqref{proof_function_cohesive} and \eqref{proof_ET_cohesive}, we can apply the same logic for the proof of Theorem \ref{operator main result} and obtain the result.

(2) By the similar modifications in (1) in the proof for (2) of Theorem \ref{operator main result} replacing ($f$-IER) with ($f$-IECoR), we can obtain the result.
\end{proof}

\section{Proofs in Section \ref{sec communication game}}\label{Appendix_proof_communication}
To prove Theorem \ref{f-ESS network 1}, the following Lemma is useful.

\begin{lemma}\label{E+FDS}
Suppose that a solution $\varphi$ satisfies (E) and ($f$-FDS) for some $f\in \mathcal{F}_{\mathscr{G}}$.
Then, for any $(v,g) \in \mathcal{V}^N \times \mathscr{G}^N$ and any $C \in N/g$, it satisfies
\[
\sum_{i \in C}\varphi_i(v,g)=\sum_{i \in C}f_i(v,g)+\frac{|C|}{n}\left( v(N)-\sum_{k \in N}f_k(v,g)  \right).
\]
\end{lemma}

\begin{proof}
Fix $f \in \mathcal{F}_{\mathscr{G}}$ and suppose that $\varphi$ satisfies (E) and ($f$-FDS).
Take any $(v,g) \in \mathcal{V}^N \times \mathscr{G}^N$ and any $C \in N/g$.
Then, by (E), we have
\[
\sum_{i \in C}\left(\varphi_i(v,g)-f_i(v,g) \right)+\sum_{C' \in N/g; C' \neq C}\sum_{j \in C'}\left(\varphi_j(v,g)-f_j(v,g) \right)=v(N)-\sum_{k \in N}f_k(v,g).
\]
Moreover, by ($f$-FDS), for each $C'\in N/g$ with $C' \neq C$, it satisfies
\[
\left( \sum_{j \in C'}\varphi_j(v, g)-f_j(v, g) \right)=\frac{|C'|}{|C|}\left( \sum_{i \in C}\varphi_i(v, g)-f_i(v, g) \right).
\]
Therefore, we have
\begin{eqnarray*}
&&\sum_{i \in C}\left(\varphi_i(v,g)-f_i(v,g) \right)+\sum_{C' \in N/g; C' \neq C}\sum_{j \in C'}\left(\varphi_j(v,g)-f_j(v,g) \right)\\
&=& \Bigl(1+ \sum_{C' \in N/g; C' \neq C} \frac{|C'|}{|C|}\Bigr)\left( \sum_{i \in C}\varphi_i(v, g)-f_i(v, g) \right)\\
&=& \frac{n}{|C|}\left( \sum_{i \in C}\varphi_i(v, g)-f_i(v, g) \right)\\
&=& v(N)-\sum_{k \in N}f_k(v,g)\\
&\Leftrightarrow & \sum_{i \in C}\varphi_i(v,g)=\sum_{i \in C}f_i(v,g)+\frac{|C|}{n}\left(v(N)-\sum_{k \in N}f_k(v,g)  \right).
\end{eqnarray*}
\end{proof}

\begin{proof}[Proof of Theorem \ref{f-ESS network 1}]
The if part is trivial. We show the only if part. 
Let $f\in \mathcal{F}_{\mathscr{G}}$ be a solution that satisfies (FA) and suppose that a solution $\varphi$ satisfies (E), (FA), and ($f$-FDS).
Take any $(v,g) \in \mathcal{V}^N \times \mathscr{G}^N$. 
We show the sufficiency of the result by the induction of the number of links $|g|$. 

Suppose that $|g|=0$.
Then, by Lemma \ref{E+FDS}, $\varphi(v,g)$ is uniquely determined by
\[
\varphi_i(v,g)=f_i(v,g)+\frac{1}{n}\left(v(N)-\sum_{k \in N}f_k(v,g)  \right)
\]
for any $i \in N$.

Suppose that $\varphi(v,g)$ is determined for any $|g| \le k$ with $k \in \bigl\{0,1, \ldots, \binom{n}{2}-1 \bigr\}$.
Let $g$ with $|g|=k+1$ and take any component $C \in N/g$ and $i \in C$.
Then, by (FA) and the connectedness of $C$, for any $j \neq i$, there are $i_1, \ldots, i_k \in C$ such that $i_1=i, i_k=j$,$i_{h+1} \in N_{i_h}(g)$ for any $h=1,\ldots, k-1$, and 
\[
\varphi_i(v, g)-\varphi_j(v, g)=\sum_{h=1}^{k-1}\Bigl( \varphi_{i_h}(v, g-i_hi_{h+1})-\varphi_{i_{h+1}}(v, g-i_hi_{h+1}) \Bigr).
\]
Note that the right-hand side of the above equation is a uniquely determined constant by the induction hypothesis.
Therefore, by combining Lemma \ref{E+FDS} with the above observations, we can obtain the $|C|$ linearly independent equations with $|C|$ unknown, which uniquely determines $\varphi_i(v,g)$ for any $i \in C$.
Since $C \in N/g$ is arbitrarily chosen, this argument shows that $\varphi(v,g)$ is uniquely determined for $g$ with $|g|=k+1$.
\end{proof}

\begin{proof}[Proof of Theorem \ref{efficient-fair main}]
The if part is trivial.
We show the only if part.
Let  $\Phi$ be an efficient-fair extension operator that satisfies all the axioms.

Take any $f \in \mathcal{F}_{\mathscr{G}}$, $g \in \mathscr{G}^N$, and $i \in N$. 
Fix any $j \in N\setminus\{i\}$, and consider the following solution $\hat f \in \mathcal{F}_{\mathscr{G}}$: for any
$(v,g') \in \mathcal{V}^N \times \mathscr{G}^N$ and $l\in N$, let
\[
\hat f_l(v,g')
=
\begin{cases}
f_i(v,g), & \text{if } l=i,\\
\displaystyle\sum_{k\in N\setminus\{i\}}f_k(v,g),
& \text{if } l=j,\\
0, & \text{otherwise}.
\end{cases}
\]
That is, player $i$ receives $f_i(v,g)$, the residual total
$\sum_{k\in N\setminus\{i\}}f_k(v,g)$ is assigned to player $j \neq i$, and all remaining players receive zero. 
Since $\hat f$ is independent of $g'$, it satisfies (FA).

Since $\Phi$ is an efficient-fair extension operator, $\Phi(\hat f)$
satisfies (E) and (FA). Moreover, $\Phi(\hat f)$ satisfies
($f$-FDS) by assumption. 
Hence, by
Theorem~\ref{f-ESS network 1},
\[
\Phi_i(\hat f)(v,g)
=
\hat f_i(v,g)
+
\frac{1}{n}
\left(
v(N)-\sum_{k\in N}\hat f_k(v,g)
\right).
\]
By construction,
\[
\hat f_i(\cdot,g)=f_i(\cdot,g)
\quad\text{and}\quad
\sum_{k\in N}\hat f_k(\cdot,g)
=
\sum_{k\in N}f_k(\cdot,g).
\]
Therefore, the premise of (WEES) is satisfied for $f$ and $\hat f$ at $g$.
It follows that
\[
\begin{aligned}
\Phi_i(f)(v,g)
&=
\Phi_i(\hat f)(v,g)\\
&=
f_i(v,g)
+
\frac{1}{n}
\left(
v(N)-\sum_{k\in N}f_k(v,g)
\right).
\end{aligned}
\]
Since $f$, $g$, $i$, and $v$ were arbitrarily chosen, $\Phi$ is the ESS
operator.
\end{proof}

\section{Proofs in Section \ref{sec_coalition}}\label{Appendix_proof_coalition}

\begin{proof}[Proof of Theorem \ref{f-ESS_coalition}]
The if part is trivial. We show the only if part.
Let $f \in \mathcal{F}_{\mathcal{C}}$ be a solution that satisfies (RBCC) and suppose that a solution $\varphi$ satisfies (E), (RBCC), ($f$-EGN), and ($f$-FDSC).
By the same argument as Lemma \ref{E+FDS}, we can see that
\begin{equation}
\sum_{i \in C}\varphi_i(N,v, \mathcal{P})=\sum_{i \in C}f_i(N, v,\mathcal{P})+\frac{|C|}{n} \left( v(N)-\sum_{k \in N}f_ k(N,v,\mathcal{P})   \right), \label{eq. eff coalition}
\end{equation}
for any $(N, v, \mathcal{P}) \in \mathcal{CV}$ and $C \in \mathcal{P}$. 
We show that $\varphi$ is uniquely determined for any $(N, v, \mathcal{P}) \in \mathcal{CV}$.

Take any $(N, v, \mathcal{P}) \in \mathcal{CV}$ and $C \in \mathcal{P}$.
If $|C|=1$,  for $i \in C$, $\varphi_i(N, v, \mathcal{P})$ is uniquely determined by the equation (\ref{eq. eff coalition}).
Suppose that $|C|=k$ for some $2 \le k \le n$.
Let us consider a game $(N', w)$ such that $w(S)=v(S\setminus \{n'\})$ for any $S \subseteq N'=N \cup \{n'\}$with $n' \in \mathcal{U}\setminus N$.
Note that $w|_N=v$ and $n'$ is a null player in $w|_{N'\setminus \{j\}}$ for any $j \in N$.
Let $\mathcal{P'}=(\mathcal{P}\setminus \{C\}) \cup \{C \cup \{n'\}\}$.
Then, by applying (RBCC) to $ \{C \cup \{n'\}\} \in \mathcal{P'}$ in the game $(N', w, \mathcal{P'})$ with an order $(i_1, \ldots, i_s, n', i_{s+1}, \ldots, i_{|C|})$, we have
\begin{eqnarray*}
&&\varphi_{i_1}(N'\setminus \{i_{|C|}\}, w, \mathcal{P'}_{-{i_{|C|}}})+\cdots+\varphi_{i_s}(N'\setminus \{i_{s-1}\}, w, \mathcal{P'}_{-{i_{s-1}}})+\varphi_{n'}(N'\setminus \{i_{s}\}, w, \mathcal{P'}_{-{i_{s}}})\\
&&~~~~~~+\varphi_{i_{s+1}}(N'\setminus \{n'\}, w, \mathcal{P'}_{-{n'}})+\cdots+\varphi_{i_{|C|}}(N'\setminus \{i_{|C|-1}\}, w, \mathcal{P'}_{-{i_{|C|-1}}})\\
&=&\varphi_{i_1}(N'\setminus \{i_{2}\}, w, \mathcal{P'}_{-{i_{2}}})+\cdots+\varphi_{i_s}(N'\setminus \{n'\}, w, \mathcal{P'}_{-{n'}})+\varphi_{n'}(N'\setminus \{i_{s+1}\}, w, \mathcal{P'}_{-{i_{s+1}}})\\
&&~~~~~~+\varphi_{i_{s+1}}(N'\setminus \{i_{s+2}\}, w, \mathcal{P'}_{-{i_{s+2}}})+\cdots+\varphi_{i_{|C|}}(N'\setminus \{i_1\}, w, \mathcal{P'}_{-{i_1}}).
\end{eqnarray*}

Moreover, by ($f$-EGN), the direct calculation shows that\footnote{Note that logic of this calculation is the same as the one demonstrated just after Theorem \ref{f-ESS_coalition}.}
\begin{eqnarray*}
&&\varphi_{i_{s+1}}(N, v, \mathcal{P})-\varphi_{i_s}(N, v, \mathcal{P})\\
&=&\varphi_{i_{s+1}}(N'\setminus \{n'\}, w, \mathcal{P'}_{-{n'}})-\varphi_{i_s}(N'\setminus \{n'\}, w, \mathcal{P'}_{-{n'}})\\
&=&\sum_{l=1; l \neq s}^{|C|}\left( \varphi_{i_l}(N'\setminus \{i_{l+1}\}, w, \mathcal{P'}_{-i_{l+1}})-\varphi_{n'}(N'\setminus \{i_{l+1}\}, w, \mathcal{P'}_{-i_{l+1}}) \right)\\
&&-\sum_{l=1; l \neq s+1}^{|C|}\left( \varphi_{i_l}(N'\setminus \{i_{l-1}\}, w, \mathcal{P'}_{-i_{l-1}})-\varphi_{n'}(N'\setminus \{i_{l-1}\}, w, \mathcal{P'}_{-i_{l-1}}) \right)\\
&=&\sum_{l=1; l \neq s}^{|C|}\left( f_{i_l}(N'\setminus \{i_{l+1}\}, w, \mathcal{P'}_{-i_{l+1}})-f_{n'}(N'\setminus \{i_{l+1}\}, w, \mathcal{P'}_{-i_{l+1}}) \right)\\
&&-\sum_{l=1; l \neq s+1}^{|C|}\left( f_{i_l}(N'\setminus \{i_{l-1}\}, w, \mathcal{P'}_{-i_{l-1}})-f_{n'}(N'\setminus \{i_{l-1}\}, w, \mathcal{P'}_{-i_{l-1}}) \right).
\end{eqnarray*}
Since the right-hand side is uniquely determined by the solution $f$, by choosing $i,j \in C$ for $i_s=i, i_{s+1}=j$,  $\varphi_i(N, v, \mathcal{P})-\varphi_j(N, v, \mathcal{P})$ is uniquely determined.
By the equation (\ref{eq. eff coalition}), this observation implies that $\varphi_i(N,v,\mathcal{P})$ is uniquely determined for any $i \in C \in \mathcal{P}$ with $|C| \le n$, which completes the proof.
\end{proof}

\begin{proof}[Proof of Theorem \ref{ESS coalition}]
The if part is trivial.
We show the only if part.
Let  $\Phi$ be an efficient-RBCC extension operator that satisfies all the axioms.

Take any $f\in\mathcal F_{\mathcal C}$ and $(N,v,\mathcal P)\in\mathcal{CV}$, and let $n=|N|$.
By the same argument used for ($f$-FDS), efficiency and ($f$-FDSC) imply that, for every $C\in\mathcal P$,
\begin{equation}
\label{eq:eff-coalition-operator}
\sum_{i\in C}\Phi_i(f)(N,v,\mathcal P)
=
\sum_{i\in C}f_i(N,v,\mathcal P)
+
\frac{|C|}{n}
\left(
v(N)-\sum_{k\in N}f_k(N,v,\mathcal P)
\right).
\end{equation}
We show that the gain
\[
\Phi_i(f)(N,v,\mathcal P)-f_i(N,v,\mathcal P)
\]
is the same for all players in each $C\in\mathcal P$.

If $|C|=1$, there is nothing to prove.
Suppose that $|C|\geq2$, and take any distinct $i,j\in C$.
Let $n'\in\mathcal U\setminus N$, put $N'=N\cup\{n'\}$, and define the game $w\in\mathcal V^{N'}$ by
\[
w(S)=v(S\setminus\{n'\})
\]
for every $S\subseteq N'$.
Thus, $w|_N=v$, and $n'$ is a null player in $w$ and in
$w|_{N'\setminus\{l\}}$ for every $l\in N$.
Let
\[
\mathcal P'
=
(\mathcal P\setminus\{C\})\cup\{C\cup\{n'\}\}.
\]

For every $k\in N$, put
\[
x_k=f_k(N,v,\mathcal P),
\qquad
x_{n'}=0.
\]
Consider a solution $\hat f\in\mathcal F_{\mathcal C}$ satisfying the
following two conditions:
\begin{enumerate}
\item For every $k\in N$, 
\[
\hat f_k(N,\cdot,\mathcal P)
=
f_k(N,\cdot,\mathcal P)
\]

\item For every $l\in N'$, every
$\mathcal Q\in\Pi^{N'\setminus\{l\}}$, and every
$k\in N'\setminus\{l\}$,
\[
\hat f_k
\left(
N'\setminus\{l\},
w|_{N'\setminus\{l\}},
\mathcal Q
\right)
=
x_k.
\]
\end{enumerate}
The two conditions are consistent. 
Indeed, when $l=n'$ and $\mathcal Q=\mathcal P$, Condition 2 agrees with Condition 1 at the game $v$. 
The solution $\hat f$ may be defined arbitrarily at all remaining arguments.

We first observe that $\hat f$ satisfies (RBCC-$w$).
To see this, take any $\mathcal Q'\in\Pi^{N'}$, any block
$D\in\mathcal Q'$, and any cyclic ordering of the players in $D$.
The terms evaluated at $(N',w,\mathcal Q')$ cancel from the two sides of the (RBCC) identity. 
By Condition 2, every remaining term associated with player $k$ is equal to $x_k$, independently of which neighboring player is deleted.
Hence, both sides of the identity are equal to the same sum.

Since $\Phi$ is an efficient-RBCC extension operator, $\Phi(\hat f)$ also satisfies (RBCC-$w$). 
Therefore, the gain function
\[
h_k(M,u,\mathcal Q)
:=
\Phi_k(\hat f)(M,u,\mathcal Q)
-
\hat f_k(M,u,\mathcal Q)
\]
satisfies the corresponding (RBCC) identity at $w$.

Write $C=\{i_1,\ldots,i_{|C|}\}$ so that
$i_s=i$ and $i_{s+1}=j$, where the subscripts are interpreted cyclically.
For every $d\in C\cup\{n'\}$ and
$k\in(C\cup\{n'\})\setminus\{d\}$, write
\[
h_k^{-d}
:=
h_k
\left(
N'\setminus\{d\},
w|_{N'\setminus\{d\}},
\mathcal P'_{-d}
\right).
\]
In particular, $h_k^{-n'}=h_k(N,v,\mathcal P)$ for every $k\in C$.
Apply (RBCC) to the block $C\cup\{n'\}\in\mathcal P'$ using the ordering
\[
(i_1,\ldots,i_s,n',i_{s+1},\ldots,i_{|C|}).
\]
The resulting identity can be rewritten as
\begin{equation}
\label{eq:gain-cycle}
h_{i_{s+1}}^{-n'}-h_{i_s}^{-n'}=
\sum_{\substack{l=1\\l\neq s}}^{|C|}
\left(
h_{i_l}^{-i_{l+1}}-h_{n'}^{-i_{l+1}}
\right)-
\sum_{\substack{l=1\\l\neq s+1}}^{|C|}
\left(
h_{i_l}^{-i_{l-1}}-h_{n'}^{-i_{l-1}}
\right).
\end{equation}
In each term on the right-hand side of
\eqref{eq:gain-cycle}, player $n'$ remains a null player and belongs to
the same block as the relevant player $i_l$. Therefore,
($f$-EGN) implies that every difference in square brackets is equal
to zero. It follows that
\begin{equation}
\label{eq:equal-gains-within-coalition}
h_i(N,v,\mathcal P)=h_j(N,v,\mathcal P).
\end{equation}
Since $i$ and $j$ were arbitrarily chosen from $C$, the gains of
$\Phi(\hat f)$ relative to $\hat f$ are equal within $C$.

By Condition 1, for every $k\in N$,
\[
\hat f_k(N,\cdot,\mathcal P)
=
f_k(N,\cdot,\mathcal P)
\quad\text{and}\quad
\sum_{l\in N}\hat f_l(N,\cdot,\mathcal P)
=
\sum_{l\in N}f_l(N,\cdot,\mathcal P).
\]
Hence, (WEES) implies
\[
\Phi_k(\hat f)(N,\cdot,\mathcal P)
=
\Phi_k(f)(N,\cdot,\mathcal P).
\]
Combining this equality with
\eqref{eq:equal-gains-within-coalition}, we obtain
\[
\Phi_i(f)(N,v,\mathcal P)-f_i(N,v,\mathcal P)
=
\Phi_j(f)(N,v,\mathcal P)-f_j(N,v,\mathcal P)
\]
for every $i,j\in C$.

It now follows from \eqref{eq:eff-coalition-operator} that
\[
\Phi_i(f)(N,v,\mathcal P)
=
f_i(N,v,\mathcal P)
+
\frac{1}{n}
\left(
v(N)-\sum_{k\in N}f_k(N,v,\mathcal P)
\right)
\]
for every $i\in C$. Since $C\in\mathcal P$, $f$, and
$(N,v,\mathcal P)$ were arbitrarily chosen, $\Phi$ is the ESS operator,
which completes the proof.

\end{proof}

\section{Independence of axioms of Theorems \ref{efficient-fair main} and \ref{ESS coalition}}\label{appendix: independence}

\subsection{On Theorem \ref{efficient-fair main}}
\label{appendix: independence_network}

\subsubsection*{Independence of (WEES)}
\underline{The case $n\geq3$.}
Fix distinct players $1,2,3\in N$, a game
$\bar v\in\mathcal V^N$, and the network
\[
\bar g=\{12\}.
\]
Thus, the only non-singleton component of $\bar g$ is $\{1,2\}$.
For every $f\in\mathcal F_{\mathscr G}$, define
\[
q_{\bar v}(f)
:=
\begin{cases}
1,&\text{if $f$ does not satisfy (FA-$v$)}\\
0,&\text{otherwise},
\end{cases}
\]
and put
\[
t(f)
:=
q_{\bar v}(f)f_2(\bar v,\bar g),
\]
and define
\[
\Phi_i(f)(v,g)
=
ESS_i(f)(v,g)
+
\textbf{1}_{\{(v,g)=(\bar v,\bar g)\}}
\begin{cases}
t(f),&i=1,\\
-t(f),&i=2,\\
0,&i\notin\{1,2\}.
\end{cases}
\]

The additional term is a zero-sum transfer within the component
$\{1,2\}$. 
It therefore preserves efficiency and the total gain of every component. 
Hence, $\Phi(f)$ satisfies ($f$-FDS) for every $f$.

If $f$ satisfies (FA-$v$), then
$q_{\bar v}(f)=0$. 
For every $v\neq\bar v$, the additional term is also zero. 
Thus, whenever local fairness must be preserved, $\Phi(f)$ coincides with $ESS(f)$. 
Consequently, $\Phi$ is an efficient-fair extension operator.

It remains to show that (WEES) fails. 
Choose pairwise distinct constants $c_4,\ldots,c_n$ that are also distinct from
$0,\pm1,\pm2$. 
For every $v\in\mathcal V^N$, define
\[
f(v,\bar g)
=
(0,1,-1,c_4,\ldots,c_n)
\]
and
\[
f'(v,\bar g)
=
(0,2,-2,c_4,\ldots,c_n).
\]
When $n=3$, the coordinates $c_4,\ldots,c_n$ are omitted. 
At
$(\bar v,\varnothing)$, set
\[
f_1(\bar v,\varnothing)
=
f_2(\bar v,\varnothing)
=
f'_1(\bar v,\varnothing)
=
f'_2(\bar v,\varnothing)
=
0.
\]
The remaining values of $f$ and $f'$ may be defined arbitrarily.
Both solutions violate (FA-$v$), and hence
\[
q_{\bar v}(f)=q_{\bar v}(f')=1.
\]
For player $1$,
\[
f_1(\cdot,\bar g)=f'_1(\cdot,\bar g)=0
\]
and
\[
\sum_{k\in N}f_k(\cdot,\bar g)
=
\sum_{k\in N}f'_k(\cdot,\bar g)
=
\sum_{k=4}^n c_k.
\]
Thus, the premise of (WEES) is satisfied. However,
\[
t(f)=1
\quad\text{and}\quad
t(f')=2.
\]
The ESS terms for player $1$ coincide, whereas the additional terms do not. 
Therefore,
\[
\Phi_1(f)(\bar v,\bar g)
\neq
\Phi_1(f')(\bar v,\bar g),
\]
so (WEES) is violated.

\medskip
\underline{The case $n=2$.}
Let $N=\{1,2\}$, fix $\bar v\in\mathcal V^N$, and let
$\bar g=\{12\}$. 
Define
\[
q_{\bar v}(f)
:=
\begin{cases}
1,&\text{if $f$ does not satisfy (FA-$v$),}\\
0,&\text{otherwise,}
\end{cases}
\]
and
\[
d(f)
:=
f_1(\bar v,\bar g)-f_2(\bar v,\bar g).
\]
Consider the operator
\[
\begin{aligned}
\Phi_1(f)(v,g)
&=
ESS_1(f)(v,g)
+
1_{\{(v,g)=(\bar v,\bar g)\}}q_{\bar v}(f)d(f),\\
\Phi_2(f)(v,g)
&=
ESS_2(f)(v,g)
-
1_{\{(v,g)=(\bar v,\bar g)\}}q_{\bar v}(f)d(f).
\end{aligned}
\]

The additional terms form a zero-sum transfer, so efficiency is preserved.
If $f$ satisfies (FA-$v$), then
$q_{\bar v}(f)=0$, while at every $v\neq\bar v$ the additional terms are zero.
Thus, the operator preserves local fairness. Moreover, ($f$-FDS) is vacuous at the connected network $\bar g$, and the operator coincides with $ESS$ at the empty network. Hence, the operator satisfies all the remaining requirements.

To see that (WEES) fails, choose $f$ and $f'$ such that
\[
f(\cdot,\bar g)
=
f'(\cdot,\bar g)
=
(1,0).
\]
At the empty network, set
\[
f(\bar v,\varnothing)=(0,-1)
\quad\text{and}\quad
f'(\bar v,\varnothing)=(0,0).
\]
The solution $f$ satisfies $(FA\text{-}\bar v)$, because deleting
the link $12$ changes both players' payoffs by $1$. 
In contrast, $f'$ violates $(FA\text{-}\bar v)$. Therefore,
\[
q_{\bar v}(f)=0,
\qquad
q_{\bar v}(f')=1,
\qquad
d(f)=d(f')=1.
\]
The two benchmarks have the same player-$1$ function and the same total
at $\bar g$, but
\[
\Phi_1(f)(\bar v,\bar g)
\neq
\Phi_1(f')(\bar v,\bar g).
\]
Thus, (WEES) is violated also when $n=2$.

\subsubsection*{Independence of ($f$-FDS)}
Suppose that $n\geq2$, and define the equal-division operator by
\[
\Phi_i^{ED}(f)(v,g)
:=
\frac{v(N)}{n}
\]
for every $f\in\mathcal F_{\mathscr G}$, every
$(v,g)\in\mathcal V^N\times\mathscr G^N$, and every $i\in N$.

The output is efficient because
\[
\sum_{i\in N}\Phi_i^{ED}(f)(v,g)=v(N).
\]
Moreover, $\Phi^{ED}(f)$ is independent of the network and therefore satisfies (FA-$v$) for every $v$, regardless of whether the benchmark solution $f$ does. 
Hence, $\Phi^{ED}$ is an efficient-fair extension operator. 
Since its output does not depend on $f$, it also satisfies (WEES).

To see that ($f$-FDS) fails, let $g^0$ be the empty network, choose distinct players $1,2\in N$, and take a benchmark $f$ and a game $v$ such that
\[
f_1(v,g^0)=1
\quad\text{and}\quad
f_2(v,g^0)=0.
\]
Since $\{1\}$ and $\{2\}$ are components of $g^0$, their average gains relative to $f$ are
\[
\Phi_1^{ED}(f)(v,g^0)-f_1(v,g^0)
=
\frac{v(N)}{n}-1
\]
and
\[
\Phi_2^{ED}(f)(v,g^0)-f_2(v,g^0)
=
\frac{v(N)}{n},
\]
respectively. 
These gains are different. Thus, $\Phi^{ED}(f)$ does not satisfy ($f$-FDS).

\subsection{On Theorem \ref{ESS coalition}}
\label{appendix: independence_coalition}
For notational convenience, let $F_f(N,v,\mathcal P):=\sum_{k\in N}f_k(N,v,\mathcal P)$.

\subsubsection*{Independence of (WEES)}

Identify three players $1<2<3$ in $\mathcal U$, and let $N_0=\{1,2,3\}$ and  $\mathcal{P}_0=\bigl\{\{1,2\},\{3\}\bigr\}$.
Let $\bar v$ be a game on $N_0$ in which neither player $1$ nor player $2$ is null; for example,
\[
\bar v(S)=|S\cap\{1,2\}|.
\]

Define $f^\ast\in\mathcal F_{\mathcal C}$ as follows. 
For every finite nonempty $M\subset\mathcal U$, every game $w$ on $M$, and the grand-coalition partition $\{M\}$, let
\[
f_i^\ast(M,w,\{M\})
=
\begin{cases}
1,&i=\min M,\\
0,&i\neq\min M.
\end{cases}
\]
At $(N_0,\cdot,\mathcal P_0)$, set
\[
f^\ast(N_0,\cdot,\mathcal P_0)=(0,1,-1),
\]
and set $f^\ast=0$ at all remaining coalition structures.

Now define
\[
\Phi_i^\ast(f)(N,v,\mathcal P)
:=
ESS_i(f)(N,v,\mathcal P)
+
H_i^\ast(f)(N,v,\mathcal P),
\]
where
\[
H^\ast(f)(N,v,\mathcal P)
:=
\begin{cases}
(1,-1,0),
&\text{if $f=f^\ast$ and
$(N,v,\mathcal P)=(N_0,\bar v,\mathcal P_0)$,}\\
0,&\text{otherwise.}
\end{cases}
\]

First, $f^\ast$ violates (RBCC-$w$) for every finite $M\subset\mathcal U$ with $|M|\geq3$ and every game $w$ on $M$.
Indeed, write $M=\{i_1, \ldots, i_m\}$ and consider the grand-coalition partition with this cyclic ordering $i_1<\cdots<i_m$. 
After the current-game terms cancel, (RBCC) would require
\[
\begin{aligned}
&\sum_{l=1}^{m}
f_{i_l}^\ast
\bigl(
M\setminus\{i_{l-1}\},
w,
\{M\setminus\{i_{l-1}\}\}
\bigr)\\
={}&
\sum_{l=1}^{m}
f_{i_l}^\ast
\bigl(
M\setminus\{i_{l+1}\},
w,
\{M\setminus\{i_{l+1}\}\}
\bigr).
\end{aligned}
\]
The left-hand side is $2$: its nonzero terms correspond to $i_1$ after deleting $i_m$ and $i_2$ after deleting $i_1$. 
The right-hand side is $1$: its only nonzero term corresponds to $i_1$ after deleting $i_2$. Thus, the required identity fails.

The component $H^\ast$ is a zero-sum transfer within the block $\{1,2\}$. 
Hence, it preserves efficiency and ($f$-FDSC) holds. 
At the exceptional point, neither player $1$ nor player $2$ is null in $\bar v$, while the other block is a singleton. 
Therefore, ($f$-EGN) imposes no restriction there. At every other point, $\Phi^\ast$ coincides with $ESS$.

If $f\neq f^\ast$, the operator is the ESS operator and preserves every local (RBCC) property of $f$. 
If $f=f^\ast$ and $|N|\geq3$, the preceding calculation shows that $f^\ast$ never satisfies (RBCC-$v$), so the local-preservation implication is vacuous.
If $|N|\leq2$, the $H$ is zero. 
Thus, $\Phi^\ast$ is an efficient-RBCC extension operator satisfying ($f$-EGN) and ($f$-FDSC).

Finally, choose $f'\neq f^\ast$ such that
\[
f'(N_0,\cdot,\mathcal P_0)
=
f^\ast(N_0,\cdot,\mathcal P_0),
\]
and modify $f'$ only at an unrelated player set or coalition structure.
For player $1$,
\[
f_1^\ast(N_0,\cdot,\mathcal P_0)
=
f'_1(N_0,\cdot,\mathcal P_0)
\]
and
\[
\sum_{k\in N_0}f_k^\ast(N_0,\cdot,\mathcal P_0)
=
\sum_{k\in N_0}f'_k(N_0,\cdot,\mathcal P_0)
=
0.
\]
Thus, the premise of (WEES) is satisfied.
 However,
\[
\Phi_1^\ast(f^\ast)(N_0,\bar v,\mathcal P_0)
=
ESS_1(f^\ast)(N_0,\bar v,\mathcal P_0)+1,
\]
whereas
\[
\Phi_1^\ast(f')(N_0,\bar v,\mathcal P_0)
=
ESS_1(f')(N_0,\bar v,\mathcal P_0).
\]
The two ESS terms coincide, so (WEES) is violated.

\subsubsection*{Independence of ($f$-EGN)}

Fix $\lambda\in\mathbb{R}\setminus\{0\}$, and define
\[
H_i^G(f)(N,v,\mathcal P)
:=
\begin{cases}
\displaystyle
\lambda
\left(
f_i(N,v,\{N\})
-
\frac{F_f(N,v,\{N\})}{|N|}
\right),
&\mathcal P=\{N\},\\[3mm]
0,&\mathcal P\neq\{N\}.
\end{cases}
\]
Consider the operator
\[
\Phi_i^G(f)(N,v,\mathcal P)
:=
ESS_i(f)(N,v,\mathcal P)
+
H_i^G(f)(N,v,\mathcal P).
\]

Since $\sum_{i \in N}H^G_i(f)(N,v, \mathcal{P})=0$, efficiency is preserved. 
It is nonzero only at the grand-coalition partition, where ($f$-FDSC) is vacuous. 
Moreover, the correction for player $i$ depends only on $f_i$ and the total $F_f$, so (WEES) holds.

We next verify local (RBCC) preservation.
If $|N|=1$, (RBCC) is vacuous.
Hence, suppose that $|N|\geq2$.
Suppose that $f$ satisfies (RBCC-$v$). 
If $\mathcal P\neq\{N\}$, every non-vacuous (RBCC) identity involves only non-grand-coalition partitions, and therefore all the $H^G$-terms are zero.
If $\mathcal P=\{N\}$, enumerate $N=\{i_1,\ldots,i_n\}$ cyclically, with $i_0=i_n$ and $i_{n+1}=i_1$. 
The $H^G$-terms evaluated at $N$ cancel from the two sides of the $(RBCC)$ identity. 
The terms involving $f$ satisfy the required equality because $f$ satisfies (RBCC-$v$), and the remaining average terms agree because
\[
\sum_{l=1}^n
\frac{
F_f(N\setminus\{i_{l-1}\},v,\{N\setminus\{i_{l-1}\}\})
}{n-1}
=
\sum_{l=1}^n
\frac{
F_f(N\setminus\{i_{l+1}\},v,\{N\setminus\{i_{l+1}\}\})
}{n-1}.
\]
Indeed, the two sides sum over the same set of deleted players. 
Hence, $\Phi^G$ is an efficient-RBCC extension operator.

To see that ($f$-EGN) fails, take $N=\{1,2\}$, $\mathcal P=\{N\}$, and the zero game $v$.
Both players are null in $v$. 
Choose a benchmark $f$ such that
\[
f_1(N,v,\mathcal P)=1
\quad\text{and}\quad
f_2(N,v,\mathcal P)=0.
\]
Since the ESS correction is common to the two players,
\[
ESS_1(f)(N,v,\mathcal P)-ESS_2(f)(N,v,\mathcal P)=1.
\]
On the other hand,
\[
H_1^G(f)(N,v,\mathcal P)-H_2^G(f)(N,v,\mathcal P)=\lambda.
\]
Consequently,
\[
\begin{aligned}
\Phi_1^G(f)(N,v,\mathcal P)-\Phi_2^G(f)(N,v,\mathcal P)
&=
1+\lambda\\
&\neq
1\\
&=
f_1(N,v,\mathcal P)-f_2(N,v,\mathcal P).
\end{aligned}
\]
Thus, ($f$-EGN) is violated.

\subsubsection*{Independence of ($f$-FDSC)}

Let $\mathcal P^0(N):=\bigl\{\{i\}:i\in N\bigr\}$ be the discrete partition of $N$. 
Fix $\lambda\in\mathbb{R}\setminus\{0\}$, and define
\[
H_i^D(f)(N,v,\mathcal P)
:=
\begin{cases}
\displaystyle
\lambda
\left(
f_i(N,v,\mathcal P^0(N))
-
\frac{F_f(N,v,\mathcal P^0(N))}{|N|}
\right),
&\mathcal P=\mathcal P^0(N),\\[3mm]
0,&\mathcal P\neq\mathcal P^0(N).
\end{cases}
\]
Consider the operator
\[
\Phi_i^D(f)(N,v,\mathcal P)
:=
ESS_i(f)(N,v,\mathcal P)
+
H_i^D(f)(N,v,\mathcal P).
\]

Since $\sum_{i \in N}H^D_i(f)(N,v, \mathcal{P})=0$, efficiency is preserved. 
Its player-$i$ component depends only on $f_i$ and the total $F_f$, so (WEES) holds.

At the discrete partition, every component is a singleton, and hence ($f$-EGN) imposes no restriction involving two distinct
players. 
At every non-discrete partition, the correction is zero.
Therefore, ($f$-EGN) holds.

For (RBCC), the correction is zero at every partition having a non-singleton component. 
If a component has at least three players, deleting one player leaves a non-singleton component, so all correction terms in its (RBCC) identity are zero. 
For a two-player component, the predecessor and successor of each player coincide, and the (RBCC) identity is tautological. 
At the discrete partition, all components are singletons. 
Thus, $H^D(f)$ satisfies (RBCC) for every $f$, and $\Phi^D$ is an efficient-RBCC extension operator.

To see that ($f$-FDSC) fails, take $N=\{1,2\}$, $\mathcal P=\mathcal P^0(N)$, and a benchmark $f$ such that, at some game $v$,
\[
f_1(N,v,\mathcal P)=1
\quad\text{and}\quad
f_2(N,v,\mathcal P)=0.
\]
Then
\[
H_1^D(f)(N,v,\mathcal P)=\frac{\lambda}{2}
\quad\text{and}\quad
H_2^D(f)(N,v,\mathcal P)=-\frac{\lambda}{2}.
\]
The ESS part gives both singleton components the same gain relative to $f$, whereas the $H^D$-terms are different. 
Hence,
\[
\Phi_1^D(f)(N,v,\mathcal P)-f_1(N,v,\mathcal P)
\neq
\Phi_2^D(f)(N,v,\mathcal P)-f_2(N,v,\mathcal P),
\]
which violates ($f$-FDSC).

%%%%%%
%%%%%%

\bibliography{bib_f-ESS}

\end{document}